\documentclass[a4paper,11pt]{article}

\usepackage[vmargin=1in]{geometry}
\usepackage{float}
\usepackage{enumitem}
\usepackage{appendix}
\usepackage{authblk}
\usepackage{xcolor}

\usepackage[utf8]{inputenc}
\usepackage[T1]{fontenc}
\usepackage[english]{babel}
\usepackage{tgpagella} 
\usepackage{xspace}
\usepackage{hyphenat}
\usepackage{csquotes}
\usepackage{comment}
\usepackage[showdeletions]{color-edits}
\usepackage[normalem]{ulem}
\usepackage{cancel}

\usepackage{amsmath,amsthm,amssymb,mathtools}
\usepackage{aliascnt} 
\usepackage{nicefrac}
\usepackage{dsfont}
\usepackage{complexity}
\usepackage{physics}
\usepackage{graphicx}
\usepackage{subcaption}
\usepackage{tikz}
\usepackage{pgfplots}
\pgfplotsset{compat=1.18}
\usepgfplotslibrary{fillbetween}
\usetikzlibrary{nfold,quantikz2}

\definecolor{compcolor}{RGB}{33,113,181} 
\definecolor{soundlog}{RGB}{107,174,214}   
\definecolor{soundpoly}{RGB}{8,81,156}

\theoremstyle{plain}
\newtheorem{theorem}{Theorem}[section] 

\newaliascnt{lemma}{theorem}
\newtheorem{lemma}[lemma]{Lemma}
\aliascntresetthe{lemma}

\newaliascnt{proposition}{theorem}

\aliascntresetthe{proposition}

\newaliascnt{corollary}{theorem}
\newtheorem{corollary}[corollary]{Corollary}
\aliascntresetthe{corollary}

\newaliascnt{procedure}{theorem}
\newtheorem{procedure}[procedure]{Procedure}
\aliascntresetthe{procedure}

\newaliascnt{definition}{theorem}
\newtheorem{definition}[definition]{Definition}
\aliascntresetthe{definition}

\theoremstyle{definition}
\newaliascnt{remark}{theorem}
\newtheorem{remark}[remark]{Remark}
\aliascntresetthe{remark}

\usepackage[breaklinks]{hyperref}
\hypersetup{
  colorlinks = true,
  linkcolor  = blue,
  citecolor  = darkcyan,
  urlcolor   = blue
}

\usepackage{cite}
\usepackage[nameinlink, capitalise]{cleveref}

\crefname{theorem}{Theorem}{Theorems}
\crefname{lemma}{Lemma}{Lemmas}
\crefname{proposition}{Proposition}{Propositions}
\crefname{corollary}{Corollary}{Corollaries}
\crefname{definition}{Definition}{Definitions}
\crefname{procedure}{Procedure}{Procedures}
\crefname{remark}{Remark}{Remarks}
\crefname{app}{Appendix}{Appendices}
\crefname{figure}{Figure}{Figures}

\definecolor{darkcyan}{rgb}{0.0,0.55,0.55}
\definecolor{darkspringgreen}{rgb}{0.09,0.45,0.27}
\definecolor{amethyst}{rgb}{0.6, 0.4, 0.8}

\usepackage{specific_symbols}

\title{\textbf{The complexity of Stoquastic Sparse Hamiltonians}}
\author[1]{Alex B. Grilo}
\author[1,2]{Marios Rozos}
\affil[1]{Sorbonne Universit\'e, CNRS and LIP6, France}
\affil[2]{National Technical University of Athens}
\date{}

\begin{document}

\maketitle

\begin{abstract}
Despite having an unnatural definition, $\StoqMA$ plays a central role in Hamiltonian complexity, e.g., in the classification theorem of the complexity of Hamiltonians by Cubitt and Montanaro (SICOMP 2016). Moreover, it lies between the two randomized extensions of $\NP$, $\MA$ and $\AM$. Therefore, understanding the exact power of $\StoqMA$ (and hopefully collapsing it with more natural complexity classes) is of great interest for different reasons. In this work, we take a step further in understanding this complexity class by showing that the Stoquastic Sparse Hamiltonians problem ($\StoqSH$) is in $\StoqMA$. Since Stoquastic Local Hamiltonians are $\StoqMA$-hard, this implies that $\StoqSH$ is $\StoqMA$-complete. We complement this result by showing that the separable version of $\StoqSH$ is $\StoqMA(2)$-complete, where $\StoqMA(2)$ is the version of $\StoqMA$ that receives two unentangled proofs.  
\end{abstract}

\vspace{0.8cm}

\section{Introduction}

Stoquastic Hamiltonians and $\StoqMA$ bridge a surprising connection between quantum Hamiltonian complexity and fundamental questions in classical complexity theory.

We recall that in the $k$-local Hamiltonian problem, we are given the description of a Hamiltonian $H= \sum_{j \in [m]} H_j$, where each term $H_i$ acts non-trivially on $k$ qubits of an $n$-qubit system, and two numbers $\alpha$ and $\beta$, and the task is to decide whether the ground-state energy of $H$ is below $\alpha$ or above $\beta$, promised that one of the two cases holds. This is the quantum analogue of classical constraint satisfaction problems and was shown by Kitaev to be complete for $\QMA$ \cite{Kit99}, the quantum analog of $\NP$. The $\QMA$-completeness of the local Hamiltonian problem presents a promising way of studying key problems in condensed matter physics, such as entanglement in ground states of Hamiltonians, via a complexity-theoretic lens.

Given the importance of this problem, there has been much effort in understanding the complexity of the local Hamiltonian problem for subfamilies of Hamiltonians of practical interest~\cite{2localkempe, 1dlineaharonov,  hubbardchilds, BH17, piddockmontanaro}. In~\cite{BBT06,BDOT08,BT10}, the authors study the complexity of the local {\em stoquastic} Hamiltonian problem, in which each term $H_j$ has nonpositive off-diagonal elements. These Hamiltonians have ground states with non-negative amplitudes, which can be interpreted as a probability distribution, and thus they do not suffer from the ``sign problem'', being more responsive to classical methods such as Monte Carlo simulations.

\cite{BBT06,BDOT08} showed that the stoquastic local Hamiltonian problem with inverse polynomial $\beta-\alpha$ is complete for $\StoqMA$. $\StoqMA$ is defined to be the complexity class where we have a quantum proof, a classical reversible circuit, and a measurement in the Hadamard basis to decide acceptance/rejection. While such a complexity class may seem somewhat artificial, it plays a key role in the classification theorem on the complexity of local Hamiltonian problems~\cite{classiciation2localHamiltonians}. From a complexity theory perspective, $\StoqMA$ lies between the two randomized generalizations of $\NP$:
\begin{align}
    \NP \subseteq \MA \subseteq \StoqMA \subseteq \AM.
\end{align}

We note that \cite{BT10} also showed that the stoquastic local Hamiltonian problem with $\alpha = 0$ and inverse polynomial $\beta$ (also known as {\em frustration-free} Hamiltonian), is \MA{}-complete. \cite{AG19} shows that such a problem actually sits in $\NP$ if $\beta = \Omega(1)$. 
Thus, understanding stoquastic Hamiltonians and $\StoqMA$ provides an interesting path toward progress in the fundamental questions  $\NP$ vs. $\MA$ and $\NP$ vs. $\AM$.

One direction to shed new light on $\StoqMA$ is to identify new problems that are complete for this class. Previous progress in this direction has focused on adding more structure to stoquastic local Hamiltonians. These include certain restricted Hamiltonian models -- such as the Transverse field Ising Model on degree-3 graphs, \cite{BH17}, and the Antiferromagnetic Transverse field Ising Model, \cite{piddockmontanaro} -- or geometrically local stoquastic Hamiltonians, such as on 2D-lattices or 1D-lines, \cite{raza2025complexitygeometricallylocalstoquastic}, or on spatially sparse graphs, \cite{waite2025complexitylocalstoquastichamiltonians}.

In this work, we take a complementary approach by proving $\StoqMA$-completeness of a problem with \emph{less} structure. In particular, we examine \emph{sparse stoquastic Hamiltonians}. A Hamiltonian $H$ is $d$-sparse, where $d=\poly(n)$, if each row of $H$ has at most $d$ nonzero entries, which either can be computed efficiently with a circuit or we assume that we have access to them via an oracle. Sparse Hamiltonians form a more general family of Hamiltonians than local Hamiltonians -- in fact, every $k$-local Hamiltonian with $m$ terms is $2^km$-sparse. The main focus of this paper is understanding the computational complexity of stoquastic sparse Hamiltonians ($\StoqSH$).

\subsection{Results}

Our main contribution is showing that $\StoqSH$ is $\StoqMA$-complete.

\begin{theorem}\label{thm-intro:StoqSH}
    $\StoqSH_{\alpha, \beta} \text{ is } \StoqMA_{a, b}\text{-complete}$, where $a-b\geq 1/\poly(n)$.
\end{theorem}

The hardness direction follows immediately from the $\StoqMA$-hardness of the local variant~\cite{BDOT08, BBT06}, since every $k$-local stoquastic Hamiltonian is a special case of a sparse stoquastic Hamiltonian. Our focus here is, therefore, to establish containment in $\StoqMA$.

The core difficulty lies in constructing a stoquastic verification procedure whose acceptance probability is a linear function of the energy expectation value $\ev{H}{\psi}$ of a given witness state $\ket{\psi}$. In the general (non-stoquastic) case, \cite{ATS03} show that the sparse Hamiltonian problem is in $\QMA$ by showing that one can use the Phase Estimation algorithm \cite{kitaevbook, NC10} to estimate the energy of a sparse Hamiltonian. This verification approach, however, is inherently quantum and cannot be leveraged by stoquastic verifiers. Moreover, the techniques of \cite{BDOT08} to put stoquastic {\em local} Hamiltonians in $\StoqMA$ do not extend to the sparse setting, since their approach consists of converting each local term of the Hamiltonian into a $\ketbra{+}$ projection rotated by a classically reversible circuit.

We resolve this obstacle by exploiting the decomposition of a sparse Hamiltonian into $1$-sparse terms \cite{BCCKS14}. For each term in this decomposition, we construct a verifier that runs a ``Hadamard-like" test using the witness state and the oracle that outputs the Hamiltonian matrix elements, and whose acceptance probability depends on $\ev{H}{\psi}$. 

A crucial ingredient of our verification procedure is a generalization of $\StoqMA$ veri\-fiers that allows measurements on a polynomial number of qubits. We show that such a generalization does not increase the power of the class, incurring just a loss in the completeness and soundness parameters. This generalization of $\StoqMA$ defines a class called $\gStoqMA$ and showcases this interesting property of multi-qubit measurements in  $\StoqMA$, which had not been considered before and may be of independent interest.

\begin{theorem}\label{thm-intro:gStoqMA=StoqMA}
    $\StoqMA_{a,b} \subseteq \gStoqMA_{a,b} \subseteq \StoqMA_{a',b'}$, where $ a' = \frac{1+ a}{2}$ and $b' = \frac{1 + b}{2}$.
\end{theorem}

We summarize the complexity landscape of local and sparse Hamiltonians in \cref{fig:results}, in the general and in the stoquastic setting.
\begin{figure}[H]
\centering
    \begin{tikzpicture}[
        box/.style={minimum width=4cm, minimum height=2.5cm, align=center},
        general_local/.style={box, fill=red!15},
        general_sparse/.style={box, fill=red!30},
        stoquastic_local/.style={box, fill=darkcyan!15},
        stoquastic_sparse/.style={box, fill=darkcyan!30}
    ]
    
    \begin{scope}
    \clip[rounded corners=8pt] (-2,-1.25) rectangle (6,3.75);
    
    \node[general_local] at (0,2.5)
    { \textbf{General Local}\\
    $\QMA$-complete\\
    \cite{kitaevbook} };
    
    \node[general_sparse] at (4,2.5)
    { \textbf{General Sparse}\\
    $\QMA$-complete\\
    \cite{ATS03} };
    
    \node[stoquastic_local] at (0,0)
    { \textbf{Stoquastic Local}\\
    $\StoqMA$-complete\\
    \cite{BBT06} };
    
    \node[stoquastic_sparse] at (4,0)
    { \textbf{Stoquastic Sparse}\\
    $\StoqMA$-complete\\
    \emph{[this work]} };
    
    \end{scope}
    
    \node[rotate=90] at (-2.5, 2.5) {\textbf{General}};
    \node[rotate=90] at (-2.5, 0) {\textbf{Stoquastic}};
    \node at (0,-1.8) {\textbf{Local}};
    \node at (4,-1.8) {\textbf{Sparse}};
    
    \end{tikzpicture}
    \caption{Complexity classification of Hamiltonian problems under locality, sparsity, and stoquasticity constraints.}
    \label{fig:results}
\end{figure}

Now, let us discuss some direct implications of \cref{thm-intro:StoqSH}.  Ioannou et al. \cite{Ioannou} study the distinction between termwise stoquastic Hamiltonians and globally stoquastic ones. A stoquastic $O(1)$-local Hamiltonian $H$ is called \emph{globally stoquastic}, if the off-diagonal entries $H$ are nonpositive, and \emph{termwise stoquastic}, if there exists a decomposition into local stoquastic terms. Most complexity-theoretic results about stoquasti\-city (e.g. $\StoqMA$-completeness of \cite{BBT06}) focused only on the termwise case. In \cite{Ioannou} they show that for a fixed locality parameter $k$, globally stoquastic Hamiltonians are a strictly larger class of Hamiltonians, so the corresponding local Hamiltonian problem is $\StoqMA$-hard, and they show that it is also contained in $\StoqMA$ with an approach similar to \cite{BBT06}. Notice that a globally stoquastic Hamiltonian is also sparse, so our verification procedure provides an alternative proof for the following result.

\begin{corollary}
    The local Hamiltonian problem for globally stoquastic $k$-local Hamiltonians is $\StoqMA$-complete.
\end{corollary}

We then proceed with our contribution about the complexity of the Separable Stoquastic Sparse Hamiltonian problem ($\SepStoqSH$). In this problem, we want to decide if the {\em bipartite product state} that minimizes the energy of the Hamiltonian has value below $\alpha$ or above $\beta$.

We define $\StoqMA(k)$, a version of $\StoqMA$ with $k$ unentangled proofs, and we show that $\SepStoqSH$ is complete for $\StoqMA(2)$. While membership follows directly from our verification procedure that we developed above, the proof for hardness requires a careful adaptation of the circuit-to-Hamiltonian construction for separable sparse Hamiltonians \cite{CS12} to the stoquastic setting, together with a perturbative analysis to compensate for the lack of error amplification in $\StoqMA(2)$, as was done in \cite{BBT06}. \Cref{fig:hamiltonian_complexity} summarizes the known results on the complexity of the Separable Hamiltonian problem in the general and the stoquastic setting.

\begin{figure}[htpb]
    \centering
    \begin{tikzpicture}[
        box/.style={minimum width=4cm, minimum height=2.5cm, align=center},
        general_local/.style={box, fill=red!15},
        general_sparse/.style={box, fill=red!40},
        stoquastic_local/.style={box, fill=darkspringgreen!15},
        stoquastic_sparse/.style={box, fill=darkspringgreen!40}
    ]
    
    \begin{scope}
    \clip[rounded corners=8pt] (-2,-1.25) rectangle (6,3.75);
    
    \node[general_local] at (0,2.5)
    { \textbf{Separable Local}\\
    $\QMA$-complete\\
    \cite{CS12} };
    
    \node[general_sparse] at (4,2.5)
    { \textbf{Separable Sparse}\\
    $\QMA(2)$-complete\\
    \cite{CS12} };
    
    \node[stoquastic_local] at (0,0)
    { \textbf{Separable}\\ 
    \textbf{Stoquastic Local}\\
    (open problem) \\
    };
    
    \node[stoquastic_sparse] at (4,0)
    { \textbf{Separable} \\
    \textbf{Stoquastic Sparse}\\
    $\StoqMA(2)$-complete\\
    \emph{[this work]} };
    
    \end{scope}
    
    \node[rotate=90] at (-2.5, 2.5) {\textbf{General}};
    \node[rotate=90] at (-2.5, 0) {\textbf{Stoquastic}};
    \node at (0,-1.8) {\textbf{Local}};
    \node at (4,-1.8) {\textbf{Sparse}};
    
    \end{tikzpicture}
    \caption{Complexity classification of Hamiltonian problems with separable witness under locality, sparsity, and stoquasticity constraints.}
\label{fig:hamiltonian_complexity}
\end{figure}

Finally, we examine whether allowing multiple Merlins increases the power of stoquastic Merlin-Arthur proof systems. We show that for suitable choices of the completeness and soundness parameters, $k$ Merlins can be simulated by two Merlins. More precisely, for $k = \poly(n)$ we show that $\StoqMA_{c,s}(k) \subseteq \StoqMA_{c',s'}(2)$, for certain values of $c,s$. The proof follows the protocol of \cite{HM10} for $\QMA(k) = \QMA(2)$: the two Merlins have to send Arthur the same state of the form $\ket{\psi_1} = \ket{\psi_2} = \ket{\phi_1} \otimes \dots \otimes \ket{\phi_k}$. Then, Arthur randomly chooses between  two actions: either he runs the original verification on either $\ket{\psi_1}$ or $\ket{\psi_2}$, or performs the product test of \cite{HM10} on the two proofs. The product test consists of applying swap tests between the corresponding registers of the two proofs. Notice that this is a valid protocol for stoquastic verifiers, since, as discussed above, they admit multi-qubit measurements. Using this approach, we get the following theorem.

\begin{theorem}\label{thm-intro:stoqma(k)=stoqma(2)}(simplified)
    $\StoqMA_{c,s}(k) \subseteq \StoqMA_{c', s'}(2)$, for certain values of $c,s, c'$ and $s'$.
\end{theorem}

The values of $c$ and $s$ that satisfy \cref{thm-intro:stoqma(k)=stoqma(2)} can be seen in \cref{fig:stoqma(k)=stoqma(2)-values}.

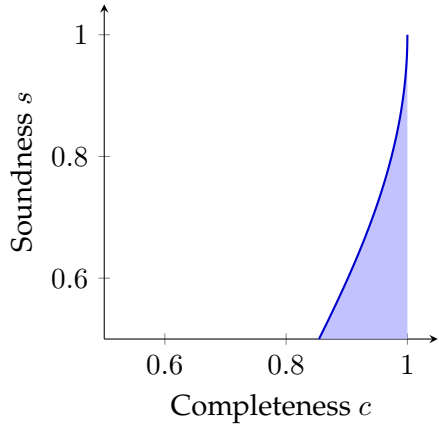
\begin{figure}[H]
  \centering
    \begin{tikzpicture}
    \begin{axis}[
        width=6cm,
        height=6cm,
        xmin=0.5, xmax=1.05,
        ymin=0.5, ymax=1.05,
        xlabel={Completeness $c$},
        ylabel={Soundness $s$},
        axis lines=left,
        samples=300,
        domain=0.5:1,
    ]
    
    \addplot[
        thick,
        color=blue!80!black,
        name path=curve
    ]
    ({(2 + (x^2 - 2*x)^2)/3}, {x});
    
    \addplot[
        draw=none,
        name path=right
    ]
    ({1}, {x});
    
    \addplot[
        blue!30,
        fill opacity=0.8
    ]
    fill between[
        of=curve and right,
    ];
    
    \end{axis}
    \end{tikzpicture}
    \caption{Region in the $(c,s)$-plane where \cref{thm-intro:stoqma(k)=stoqma(2)} holds.}
    \label[figure]{fig:stoqma(k)=stoqma(2)-values}
\end{figure}

\subsection{Related work}

Stoquastic sparse Hamiltonians have already appeared implicitly in several contexts, even though the corresponding Hamiltonian problem has not been studied directly. For example, in the proof that $\StoqMA \subseteq \AM$ \cite{BDOT08}, Bravyi et al. observe that the proof does not require the Hamiltonian to be local, but works also in the case of sparsity. Regarding adiabatic evolution, \cite{Hastings2021powerofadiabatic, gilyenHastingsVazirani} exhibited a subexponential oracle separation between classical and adiabatic quantum computation using frustrated (i.e. the ground energy is nonzero) sparse stoquastic Hamiltonians. More recently, building on the fixed-node Hamiltonian construction introduced in \cite{FixedNodeHamiltonian}, Bravyi et al. \cite{Bravyi2023rapidlymixingmarkov} showed that, given suitable guiding states, one can reduce a general Hamiltonian to a frustration-free stoquastic sparse Hamiltonian and use the latter within classical algorithms.  \cite{JiangSuccinct} subsequently used this framework to prove that the Local Hamiltonian Problem with succinctly describable ground states is $\MA$-complete, even for general Hamiltonians.

\subsection{Conclusions and open questions}

In this work, we studied the complexity of the \emph{Stoquastic Sparse Hamiltonian problem} and proved that it is $\StoqMA$-complete. Our approach relies on showing that $\StoqMA$ admits certain multi-qubit measurements, which may be of independent interest. These results naturally motivate an investigation of the power of stoquastic verifiers, when the witness is restricted to be separable. To this end, we introduced $\StoqMA(k)$, the stoquastic analogue of $\QMA(k)$, and showed that $\StoqMA(2)$ is complete for the Stoquastic Sparse Hamiltonian problem, when we have the guarantee of separable witnesses. 

This places $\StoqMA(2)$ as an interesting intermediate model between single-prover stoquastic verification and general multi-prover quantum verification, and raises several questions about the structure and limitations of stoquastic proof systems, which we mention below.

\begin{itemize}
    \item \textbf{Robustness of $\StoqMA(k)$ collapse.}
    We showed that $\StoqMA(k) \subseteq \StoqMA(2)$ for certain values of completeness and soundness parameters. An important open problem would be whether this inclusion can be made more robust under improved values of parameters, or whether tighter error bounds introduce new subtleties in the stoquastic multi-prover setting.

    \item \textbf{Separable Stoquastic Local Hamiltonians.}
    An interesting problem, which we leave as an open question, is the complexity of the \emph{Separable Stoquastic Local Hamiltonian} problem. 
    Since we do not have an algorithm for the stoquastic analogue of the Consistency of Local Density Matrices problem (which is of independent interest~\cite{stoquasticconsistencyliu}), the proof of \cite{CS12} that the Separable Local Hamiltonian problem is $\QMA$-complete cannot be directly adapted to this setting.

    \item \textbf{Upper bounds on $\StoqMA(2)$.} 
     Perhaps the most intriguing question is to determine a nontrivial upper bound on $\StoqMA(2)$. A trivial bound is $\StoqMA(2) \subseteq \QMA(2)$, but it is not clear whether this containment can be strengthened.

    This question highlights an interesting tradeoff: $\StoqMA(2)$ involves two unentangled quantum provers interacting with a highly restricted, nearly classical verifier. While it is known that $\StoqMA \subseteq \AM$ \cite{BDOT08}, and that $\AM(k) = \AM(2) = \AM$, for $k = \poly(n)$~\cite{am2}, none of these techniques seem to directly extend to $\StoqMA(2) \stackrel{?}{\subseteq} \AM$.

\end{itemize}

\subsection{Concurrent work}

At the time of submission, we became aware of an independent work \cite{liuUnentangledStoquasticMerlinArthur2026} that has as central topic the study of $\StoqMA(k)$. That work develops a broader complexity-theoretic understanding of $\StoqMA(k)$, including showing that $\StoqMA(k) \subseteq \EXP$, stron\-ger parameter regimes for simulating $\StoqMA(k)$ by $\StoqMA(2)$, and some complete problems for $\StoqMA(k)$. Our results are complementary: while their work focuses on the structural and complexity-theoretic properties of unentangled stoquastic proof systems, we provide a Hamiltonian-complexity perspective by showing that the separable stoquastic sparse Hamiltonian problem is complete for $\StoqMA(2)$. This gives a natural complete problem for $\StoqMA(2)$, analogous to the role of separable sparse Hamiltonians for $\QMA(2)$ of \cite{CS12}.

\section{Preliminaries}

In this section we introduce the basic notions and tools used throughout the paper.

\subsection{ The Local Hamiltonian Problem}

The local Hamiltonian problem is a quantum analogue of classical constraint satisfaction problems and plays a central role in quantum complexity theory, serving as the canonical complete problem for $\QMA$.

\begin{definition}[$k$-local Hamiltonian] A Hamiltonian  $H$ on $n$ qubits is a Hermitian operator on $\C^{2^n}$, i.e., a complex Hermitian matrix of dimension $2^n\times 2^n$. A Hamiltonian on $n$ qubits is called \emph{$k$-local} if it can be written as $H = \sum_{i=1}^m H_i$, where each $H_i$ acts on at most $k = O(1)$ out of the $n$ qubits.
\end{definition}

\begin{definition}[$k$-Local Hamiltonian Problem ($\LH_{\alpha, \beta}$)]
    For a $k$-local Hamiltonian $H=\sum_{j=1}^m H_j$ on $n$ qubits, where each $0 \preceq H_j \preceq I$, decide whether $\lambda_{\min{}}(H)\leq \alpha$ or $\lambda_{min{}}(H) \geq \beta$, where $\beta - \alpha \geq 1/\poly(n)$.
\end{definition}

Kitaev \cite{Kit99} has shown that for $\alpha \text{ and } \beta$ such that $\beta - \alpha = 1/\poly(n)$, the $k$-local Hamiltonian Problem is $\QMA$-complete.

\subsection{Sparse Hamiltonians}

\begin{definition}[Sparse Hamiltonian]
    A Hamiltonian $H$ on $n$ qubits is said to be \emph{$d$-sparse}, where $d = \poly(n)$, if each row and column has at most $d$ nonzero elements.
\end{definition}
\begin{definition}[Sparse Hamiltonian Problem ($\SH_{\alpha, \beta}$)]
    For a $d$-sparse Hamiltonian $H$ on $n$ qubits, where  $0 \preceq H \preceq I$, decide whether $\lambda_{\min{}}(H)\leq \alpha$ or $\lambda_{min{}}(H) \geq \beta$, where $\beta - \alpha \geq 1/\poly(n)$.
\end{definition}

A $k$-local Hamiltonian $H$ is a special case of a sparse Hamiltonian, since each term acts only on $k$ qubits and is therefore $2^k$-sparse, implying that $H$ is $2^km$-sparse, where $m$ is the total number of local terms.

Furthermore, we suppose that the Hamiltonian is accessed via a pair of oracle unitaries $O_F$ and $O_H$. The $O_F$ oracle returns the location of the $i$-th nonzero entry $y_i$ in a given row $x$,
\begin{equation*}
    O_F \ket{x,i} = \ket{x, y_i},
\end{equation*}
while the $O_H$ oracle returns the entry of the Hamiltonian $H$ for any given row $x$ and column $y$,
\begin{equation*}
    O_H \ket{x,y, z} = \ket{x,y, z \oplus H_{x,y}}.
\end{equation*}

We give now a useful lemma regarding the decomposition of $d$-sparse Hamiltonians.

\begin{lemma}[Sparse decomposition, Lemma 4.4 of \cite{BCCKS14}]\label{lemma:1-sparse-decomposition}
    If $H$ is a $d$-sparse Hamiltonian, there exists a decomposition $H = \sum_{j=1}^{d^2} H_j$, where each $H_j$ is 1-sparse and a query to any $H_j$ can be simulated with $O(1)$ queries to $H$.
\end{lemma}

\begin{remark}\label{remark:each_element_appears_once}
An important property of this decomposition is that each nonzero matrix element of the Hamiltonian is uniquely assigned to one of the terms $H_j$ of the decomposition. This will ensure that if $H$ is stoquastic, then each term $H_j$ will be stoquastic as well.
\end{remark}

For a given $1$-sparse term $H_j$ with real coefficients and a given row $x$, we denote $x\in H_j$, if there exists a nonzero element at row $x\in H_j$, and $x\not\in H_j$ otherwise.

\begin{remark}[Oracle Simulation] \label{remark:oracle_simulation}

If $x \in H_j$, we can access the location $f(x)$ and the value $H^{(j)}_{x,f(x)}$ of the unique nonzero element by querying an oracle $O_{f, H_j}$, which can be simulated with one query from $O_f$ and $O_H$:
\begin{align*}
    \ket{x, 0, 0} &\xmapsto{O_f} \ket{x, f(x), 0}\\
                  &\xmapsto{O_H} \ket*{x, f(x), H_{x,f(x)}}.
\end{align*}    
In the case of 1-sparse Hamiltonians with real elements, it holds that $f(f(x)) = x$. This means that starting from $\ket{x}$ and with some ancilla qubits, one can produce the state $\ket{f(x)}$ as follows:
\begin{align*}
    \ket{x, 0, 0} 
        &\xmapsto{O_{f, H_j}} \ket*{x, f(x), H_{x,f(x)}} \\
        &\xmapsto{\text{SWAP}} \ket*{f(x), x, H_{x,f(x)}}\\
        &\xmapsto{O_{f,H_j}} \ket*{f(x), x \oplus f(f(x)), H_{x,f(x)} \oplus H_{x,f(x)}} = \ket{f(x), 0,0},
\end{align*}
and finally, we can ignore the ancilla register.

In order to also treat the case of applying $O_{f,H_j}$ while $x\not\in H_j$, we can add an extra flag qubit that becomes $0$ in that case and then we can condition the application of SWAP and the second oracle query on that qubit. So, in this case, we have that $f(x) = x$.
\end{remark}

\subsection{StoqMA and Stoquastic Hamiltonians}

We recall now the definition of the complexity class $\StoqMA$, introduced in \cite{BBT06}, as well as the definition of Stoquastic Hamiltonians. We begin with the definition of stoquastic verifiers.

\begin{definition}[Stoquastic Verifier, \cite{BBT06}] \label{def:stoquastic-verifier}
    A \emph{stoquastic verifier} is a tuple $\cV = (n, n_w, n_0, n_+, U)$, where:
    \begin{itemize}
        \item $n$ is the number of input bits,
        \item $n_w$ is the number of input witness qubits, 
        \item $n_0$ is the number of input ancillas in the state $\ket{0}$, 
        \item $n_+$ is the number of input ancillas in the state $\ket{+}$, and
        \item $U$ is a quantum circuit on $n+n_w + n_0 + n_+$ qubits with X, CNOT, and Toffoli gates.
    \end{itemize}
    The acceptance probability of a stoquastic verifier $\cV$  on input string $x\in\Sigma^n$ and witness state $\ket{\psi}\in (\C^2)^{\otimes n_w}$ is defined as $\Pr[\cV \text{ accepts } (x, \ket{\psi})] = \ev{U^\dagger \Pi_\out U}{\psi_\inp}$, where $\ket{\psi_\inp} = \ket{x}\otimes\ket{\psi}\otimes\ket{0}^{\otimes n_0} \otimes \ket{+}^{\otimes n_+}$ is the initial state and $\Pi_\out = \ketbra{+}_1\otimes I_\text{else}$ projects the first qubit onto the state $\ket{+}$.
\end{definition}

\begin{definition}[$\StoqMA$, \cite{BBT06}]\label{def:StoqMA}
    A promise problem $\cL = (\cL_\yes, \cL_\no) \text{ is in } \StoqMA_{a,b}$ if there exists a uniform family of stoquastic verifiers, such that for any input $x\in\cL$ and for any fixed number of input bits $n$, there exist efficiently computable functions $a(n)$ and $b(n)$ and the corresponding verifier $\cV$ obeys:

    \begin{itemize}
        \item[] \emph{\textbf{Completeness}}: If $x \in \cL_\yes$, then there is some $\ket{\psi}$ s.t. $\Pr[\cV \text{ accepts } (x, \ket{\psi})] \geq a(n)$.
        \item[] \emph{\textbf{Soundness}}: If $x \in \cL_\no$, then for any $\ket{\psi}$, $\Pr[\cV \text{ accepts } (x, \ket{\psi})] \leq b(n)$.
        \end{itemize}
\end{definition}

One can observe that, due to the fact the input state is non-negative, the soundness parameter is always at least $\frac{1}{2}$ (see Remark 2.4 in \cite{AGL25}).

It is not known currently how to reduce the error probability for $\StoqMA$, and \cite{AGL25} has shown that error reduction implies $\StoqMA = \MA$, as well as that $\StoqMA$ with perfect completeness is contained in $\MA$. Therefore, one has to be careful on the parameters of completeness and soundness and how they translate on the hardness and containment proofs.

\begin{definition}[Stoquastic Hamiltonian]
    A $k$-local Hamiltonian $H = \sum_{j=1}^m H_j$ is called stoquastic in the computational basis, if for all $j\in[m]$ the off-diagonal elements of $H_j$ in this basis are nonpositive. Similarly, a sparse Hamiltonian $H$ is stoquastic in the computational basis if all the off-diagonal elements of $H$ are nonpositive.
\end{definition}
\begin{definition}[$k$-Local Stoquastic Hamiltonian Problem ($\StoqLH_{\alpha, \beta}$)]
    For a $k$-local Hamiltonian $H=\sum_{j=1}^m H_j$ on $n$ qubits, where each $H_j$ is stoquastic and $\lVert  H_j \rVert \leq \poly(n)$ for all $j\in[m]$, decide whether $\lambda_{\min{}}(H)\leq \alpha$ or $\lambda_{min{}}(H) \geq \beta$, where $\beta - \alpha \geq 1/\poly(n)$.
\end{definition}

In this work, we will focus on the stoquastic Sparse Hamiltonian, which we formulate below.

\begin{definition}[Stoquastic Sparse Hamiltonian Problem ($\StoqSH_{\alpha, \beta}$)]
    For a $d$-sparse stoquastic Hamiltonian $H$ on $n$ qubits, where $0\preceq H \preceq I$, decide whether $\lambda_{\min{}}(H)\leq \alpha$ or $\lambda_{min{}}(H) \geq \beta$, where $\beta - \alpha \geq 1/\poly(n)$.
\end{definition}

\subsection{Other useful lemmas}

We conclude this section by stating two standard lemmas that will be used throughout the paper.

 \begin{lemma}[SWAP test for mixed states, based on \cite{swaptest}]\label{lemma:swap-test}
     Let $\rho, \sigma$ be two $n$-qubit quantum states. Consider the swap test applied to $\rho\otimes \sigma$. The probability that the test accepts is
     \begin{equation*}
         \frac{1 + \Tr(\rho\sigma)}{2}.
     \end{equation*}
 \end{lemma}

\begin{lemma}[Lemma 22 of \cite{HM10}]\label{lemma:close-overlap}
Let $\ket{\psi}, \ket{\phi}$ be pure states such that $|\braket{\psi}{\phi}|^2 = 1-\varepsilon$, and let $\Pi$ a projector. Then, $|\ev{\Pi}{\psi} - \ev{\Pi}{\phi}| \leq \sqrt{\varepsilon}$.
\end{lemma}
\section{The Stoquastic Sparse Hamiltonian Problem is StoqMA-complete}\label{section:StoqSH}

\noindent In this section we prove the main theorem.

\begin{theorem}\label{thm:StoqSH}
    $\StoqSH_{\alpha, \beta} \text{ is } \StoqMA_{a, b}\text{-complete}$, where $a-b\geq 1/\poly(n)$.
\end{theorem}

The $k$-local stoquastic Hamiltonian problem is already known to be $\StoqMA$-hard \cite{BBT06}, and every $k$-local Hamiltonian is a special case of a sparse Hamiltonian. Hence, $\StoqSH$ is $\StoqMA$-hard. To complete the proof, it remains to show that $\StoqSH \in \StoqMA$.  For that, we introduce in \cref{sec:genstoqmav} a generalization of $\StoqMA$ verifiers, that does not change the computational power of the class up to an inverse polynomial decrease of the completeness/soundness gap. In \cref{sec:protocol}, we present the generalized $\StoqMA$ verifier for the sparse stoquastic Hamiltonians problem. We then analyze its correctness in \cref{sec:analysis}. 

\paragraph{Normalization and Representation Assumptions.} Throughout this section, we  assume that the stoquastic sparse Hamiltonians have real entries in $[-1, 0]$. We assume this without loss of generality: given the oracle that outputs the entries of a stoquastic $H$ such that $0 \preceq H\preceq I$, we can provide an oracle that outputs the entries of $\frac{H - I}{2}$, which only incurs a shift in the energy (for both positive and negative instances of the stoquastic sparse Hamiltonian problem). 

We also assume that the absolute value of each entry can be computed exactly as $|H_{x,y}| = k/2^\ell$, where $k \in \{0, \dots, 2^\ell\}$ and $\ell =\poly(n)$ is the maximum number of bit precision and that for a given $x$, we have access to an oracle $O_{H_j}$ that returns the numerator $k$ of the entry $|H_{x,f(x)}^{(j)}|$. 

\paragraph{Non-negative states assumptions.} In the analysis, we only consider non-negative states. This can be done without loss of generality because the groundspace of stoquastic Hamiltonians is spanned by nonnegative states. Moreover, there is always a nonnegative state that achieves the maximum acceptance probability of $\StoqMA$ verifiers.

\subsection{Generalized StoqMA verifiers} \label{sec:genstoqmav}

We show now that the computational power of $\StoqMA$ does not change if we allow the verifier to measure a polynomial number of qubits, at the cost of a change of the completeness and soundness parameters. We start with some definitions about the modifications in the measurement operator.

\begin{definition}[Generalized stoquastic verifier] \label{def:g-stoquastic verifier}
    A generalized stoquastic verifier is a stoquastic verifier, for which the acceptance probability on input string $x\in\Sigma^n$ and witness state $\ket{\psi}$ is given by  $\Pr[\cV \text{ accepts } (x, \ket{\psi})] = \ev{U^\dagger \Pi_\out U}{\psi_\inp}$, where $\Pi_\out$ projects $\ell = \poly(n)$ qubits, the first onto the $\ket{+}$ state and the rest onto some of the $\ket{+}, \ket{0} \text{ or } \ket{1}$ states.
\end{definition}

\begin{remark} \label{remark:choice_of_measurement}
The choice of measuring the first $\ell$ qubits happens without loss of generality, since if one wanted to measure any other $\ell$ qubits, they could apply SWAP gates accordingly. Notice that SWAP gates are a legal choice of gates, since each SWAP gate can be simulated with 3 CNOT gates \cite{NC10}.
\end{remark}

\begin{definition}[$\gStoqMA$]\label{def:gStoqMA}
    $\gStoqMA_{a,b}$ is the same as $\StoqMA_{a,b}$, except the verifier is a generalized stoquastic verifier for every input size $n$ with completeness parameter $a(n)$ and soundness parameter $b(n)$.
\end{definition}

We will show that with the modification in the measurement operator we remain in $\StoqMA$, but with increased completeness and soundness parameters.

\begin{theorem}\label{thm:gStoqMA=StoqMA}
    $\StoqMA_{a,b} \subseteq \gStoqMA_{a,b} \subseteq \StoqMA_{a',b'}$, where $ a' = \frac{1+ a}{2}$ and $b' = \frac{1 + b}{2}$.
\end{theorem}
\begin{proof}
    The inclusion $\StoqMA_{a,b} \subseteq \gStoqMA_{a,b}$ is immediate, since every StoqMA verifier is also a $\gStoqMA$ verifier.
    
    For the second inclusion, let $\cV_g$ be a generalized stoquastic verifier that has as input the state $\ket{\psi}$ (the witness state along the ancilla), and suppose that $\cV_g$ is measuring with respect to some fixed state, say $\ket{\phi} = \bigotimes_{i=1}^\ell \ket{\phi_i} \in \C^{\otimes \ell}$, where $\ket{\phi_i} \in \{\ket{+}, \ket{0}, \ket{1}\}$. 
    
    We can now define a stoquastic verifier $\cV$ that constructs the state $\ket{\phi}$ using $\ket{0},\ket{1}$ and $\ket{+}$ ancilla qubits, runs the circuit $\cV_g$ on $\ket{\psi}$, and then performs a SWAP test between the $\ell$ output qubits measured by $\cV_g$ and the prepared register $\ket{\phi}$. The verifer $\cV$ accepts if and only if the SWAP test accepts.
    
    Let $\ket{\Psi} = \cV_g \ket{\psi}$. The first $\ell$ qubits of $\ket{\Psi}$ need not be in a pure state, since they may be entangled with the remaining qubits. We therefore denote their reduced density matrix by 
    \begin{equation*}
        \rho_\ell = \Tr_\text{else}(\ketbra{\Psi}).
    \end{equation*}
    Then, the acceptance probability of $\cV_g$ can be written as follows:
    \begin{equation*}
        \Pr[\cV_g \text{ accepts } \ket{\psi}] = \bra{\Psi} (\ketbra{\phi} \otimes I_\text{else}) \ket{\Psi} = \Tr(\ketbra{\phi}\rho_\ell).
    \end{equation*}
    Furthermore, the acceptance probability of $\cV$ is the acceptance probability of the SWAP test applied to $\rho_\ell$ and $\ketbra{\phi}$, so from \cref{lemma:swap-test} we get
    \begin{equation*}
        \Pr[\cV \text{ accepts } \ket{\psi}] = \frac{1 + \Tr(\ketbra{\phi}\rho_\ell)}{2}.
    \end{equation*}
    Hence,
    \begin{equation*}
        \Pr[\cV\text{ accepts } \ket{\psi}] = \frac{1}{2} + \frac{1}{2}\Pr[\cV_g \text{ accepts } \ket{\psi}].
    \end{equation*}
\end{proof}

\begin{remark}
    Notice that the result above holds for a polynomial number of qubits, which will enable us to perform multi-qubit measurements with stoquastic verifiers in the verification procedure of the next section.
\end{remark}

\subsection{The verification procedure}\label{sec:protocol}
We now define the verifier $\cV$ that, with equal probability, runs one of the two veri\-fication procedures $\cV_1$ and $\cV_2$ below.

\begin{procedure}
    Consider a stoquastic verifier $\cV_1$ that has access to the oracles $O_f, O_{H_j}$ and receives a witness state $\ket{\psi} = \sum_x \alpha_x \ket{x}$.
    \begin{enumerate}
        \item Pick $j \in \{1, \dots, d^2\}$ uniformly at random.
        \item Start from the state $\ket{+}\ket{\psi}\ket{+}^{\otimes \ell}\ket{0}^{\otimes \ell}\ket{0}$.
        \item Use the oracle $O_{H_j}$ to construct the state $\sum_{x} \alpha_x \ket{+}\ket{x}\ket{c_x}$, \text{ where }:
        \begin{equation}
             \ket{c_x} = \sqrt{1 - |H^{(j)}_{x,f(x)}|} \ket{\chi'_x} \ket{0} + \sqrt{|H^{(j)}_{x,f(x)}|}\ket*{\chi_x}\ket{1},
        \end{equation}
        for some orthogonal states $\ket{\chi_x}$ and $\ket{\chi_x'}$ defined in \cref{eq:chi_x}.
        \item Use the oracle $O_{f, H_j}$ to the state $\ket{\psi}$, controlled by the first qubit. This creates the state
        \begin{align}
           & \frac{1}{\sqrt{2}}\sum_{x}\alpha_x \ket{0}\ket{x}\ket{c_x} + 
            \frac{1}{\sqrt{2}}\sum_{x}\alpha_x \ket{1}\ket{f(x)}\ket{c_x}\\
            =& \frac{1}{\sqrt{2}}\sum_{x}\alpha_x \ket{0}\ket{x}\ket{c_x} + 
            \frac{1}{\sqrt{2}}\sum_{x}\alpha_{f(x)} \ket{1}\ket{x}\ket{c_x}.
        \end{align}
        \item Perform measurement of the first qubit in the $X$ basis and of the last qubit in the $Z$ basis, and accept if the outcomes are $\ket{+}$ and $\ket{1}$ respectively.
    \end{enumerate}
\end{procedure}

\begin{procedure}
    Consider a stoquastic verifier $\cV_2$.
    \begin{enumerate}
        \item Pick $j \in \{1, \dots, d^2\}$ uniformly at random.
        \item Start from state $\ket{+}\ket{\psi}\ket{+}^{\otimes \ell} \ket{0}^{\otimes \ell}\ket{0}$.
        \item Construct the state
            \begin{equation}
                \frac{1}{\sqrt{2}}\sum_{x}\alpha_x\ket{0}\ket{x} \ket{+}^{\otimes \ell}\ket{0}^{\otimes \ell} \ket{0} + \frac{1}{\sqrt{2}}\sum_{x}\alpha_x\ket{1}\ket{x}\ket{c_x}. 
            \end{equation}
         \item Perform measurement of the last qubit in the $Z$ basis and accept if it is $0$. 
    \end{enumerate}
\end{procedure}

\subsection{Analysis of the verification procedure} \label{sec:analysis}

We analyze now the acceptance probability of $\cV_1$ and $\cV_2$. 

\begin{lemma} \label{lemma:verifier1_StoqSH_in_StoqMA}
    The probability that $\cV_1$ accepts the witness state $\ket{\psi} = \sum_{x} \alpha_x \ket{x}$ is    
    \begin{equation*}
        \Pr[\cV_1 \text{ accepts } \ket{\psi}] = \frac{1}{2d^2}\sum_{x}\alpha_x^2 |H_{x, f(x)}| + \frac{1}{2d^2}\ev{-H}{\psi}.
    \end{equation*}
\end{lemma}
    
\begin{proof} 
At first, we are going to restrict ourselves to a circuit for a single 1-sparse term $H_j$ of the decomposition of \cref{lemma:1-sparse-decomposition}.  

The stoquastic verifier $\cV_1$ receives $\ket{\psi}$ and prepares the state $\ket{+}\ket{\psi}\ket{+}^{\otimes \ell}\ket{0}^{\otimes \ell}\ket{0}$. For any fixed $x$ of the witness, due to the assumptions we have made for the entries of $H$, if $x\in H_j$, suppose that $|H^{(j)}_{x,f(x)}| = k/2^\ell$ for some $k \in \{1, \dots, 2^\ell\}$. Notice that
\begin{equation*}
        \ket{x}\ket{+}^{\otimes \ell} \ket{0}^{\otimes \ell}\ket{0}
        = \frac{1}{2^{\ell/2}} \sum_{y\in\{0,1\}^\ell} \ket{x}\ket{y}\ket{0}^{\otimes \ell}\ket{0},
    \end{equation*}
    so, by applying the oracle $O_{H_j}$, our state becomes:
    \begin{equation*}
        \frac{1}{2^{\ell/2}} \sum_{y\in\{0,1\}^\ell} \ket{x}\ket{y}\ket{k}\ket{0}.
    \end{equation*}
Using a classical reversible comparator, we can construct the state
    \begin{align*}
        &\frac{1}{2^{\ell/2}} \sum_{y \geq k} \ket{x}\ket{y}\ket{k}\ket{0} + \frac{1}{2^{\ell/2}} \sum_{0\leq y < k} \ket{x}\ket{y}\ket{k}\ket{1},
    \end{align*}
    which can be rewritten as
    \begin{equation*}
        \sqrt{1-\frac{k}{2^\ell}}\ket{x}\ket{\chi_x'}\ket{0} + \sqrt{\frac{k}{2^\ell}}\ket{x}\ket{\chi_x}\ket{1},
    \end{equation*}
    where we define
    \begin{equation}\label{eq:chi_x}
        \ket{\chi_x} = \frac{1}{\sqrt{k}}\sum_{0 \leq y < k} \ket{y}\ket{k} \text{ and } \ket{\chi_x'} = \frac{1}{\sqrt{2^\ell - k}}\sum_{y \geq k} \ket{y}\ket{k}. 
    \end{equation}
    So, we can apply this procedure coherently to get the state
    \begin{equation*}
        \sum_{x} \alpha_x \ket{+}\ket{x}\ket{c_x}.
\end{equation*}
Using the oracle $O_{f, H_j}$ conditioned on the first qubit and the procedure described in \cref{remark:oracle_simulation}, the state becomes
\begin{equation*}
    \frac{1}{\sqrt{2}}\sum_x \alpha_x \ket{0} \ket{x} \ket{c_x} + \frac{1}{\sqrt{2}} \sum_x \alpha_x \ket{1}\ket{f(x)}\ket{c_x}.
\end{equation*}
Due to the fact that $H_j$ is 1-sparse, it holds that $f(f(x)) = x$ and there exists a one-to-one correspondence between each $x$ and $f(x)$. Also, since $|H_{x, f(x)}| = |H_{f(x), x}|$, we have that $\ket{c_x} = \ket*{c_{f(x)}}$. Thus, the state above, which we denote as $\ket{\Phi}$, can be rewritten as
\begin{align*}
    \ket{\Phi} 
    &= \frac{1}{\sqrt{2}}\sum_x \alpha_x \ket{0} \ket{x} \ket{c_x} + \frac{1}{\sqrt{2}} \sum_x \alpha_{f(x)} \ket{1}\ket{x}\ket{c_x} .
\end{align*}

From \cref{thm:gStoqMA=StoqMA}, the verifier $\cV_1$ can measure the first qubit with respect to $\ket{+}$ and the final qubit with respect to $\ket{1}$, so the projector operator is $\Pi_\acc = \ketbra{+} \otimes I_\text{else} \otimes \ketbra{1}_\final$. We can observe that 
\begin{equation*}
    \bra{c_x} (I_{\text{else}} \otimes \ketbra{1}_\final) \ket{c_x} = |H_{x,f(x)}^{(j)}|
\end{equation*}
and the contribution of the first qubit for a fixed $x$ is equal to
\begin{equation*}
    \bigg\|\bra*{+}\frac{\big(\alpha_x\ket{0} + \alpha_x\ket{1}\big)}{\sqrt{2}}\bigg\|^2 = \frac{1}{4}\big(\alpha_x + \alpha_{f(x)}\big)^2,
\end{equation*}
so we get 
\begin{equation*}
    \ev{\Pi_\acc}{\Phi} = \frac{1}{4}\sum_{x\in H_j} \big(\alpha_x + \alpha_{f(x)}\big)^2\, |H_{x,f(x)}|.
\end{equation*}

Notice that we ended up only with the terms of $x\in H_j$, since if $x \not \in H_j$, $\ket{c_x}$ is at the state $\ket{0}$, and thus they do not contribute to the measurement.

Thus, the acceptance probability is
\begin{align*}
    \Pr[\cV_1 \text{ accepts } \ket{\psi} \mid H_j] 
    &= \ev{\Pi_\acc}{\Phi}\\
    &= \frac{1}{4}\sum_{x\in H_j}(\alpha_x^2  + \alpha_{f(x)}^2) |H_{x,f(x)}| + \frac{1}{2}\sum_{x\in H_j}\alpha_x\alpha_{f(x)} |H_{x,f(x)}| \\
    & = \frac{1}{2}\sum_{x\in H_j}\alpha_x^2 |H_{x,f(x)}| + \frac{1}{2}\ev{-H_j}{\psi}.
\end{align*}
where in the second equality we use the one-to-one correspondence between $x$ and $f(x)$.

In the general case, the verifier $\cV_1$ chooses an integer $j$ from $1$ to $d^2$ uniformly at random. Thus, by averaging over all $j$, the total probability of $\cV_1$ accepting is 
\begin{align*}
    \Pr[\cV_1 \text{ accepts } \ket{\psi}] 
        &= \sum_{j=1}^{d^2} \frac{1}{d^2}\Pr[\cV_1 \text{ accepts} \ket{\psi} \mid H_j]\\
        &= \frac{1}{2d^2} \sum_{j=1}^{d^2} \sum_{x\in H_j}\alpha_x^2 |H_{x,f(x)}|
        + \frac{1}{2d^2}\sum_{j=1}^{d^2} \mel{\psi}{-H_j}{\psi}\\
        &= \frac{1}{2d^2} \sum_{x} \alpha_x^2 |H_{x,f(x)}|+ \frac{1}{2d^2}\ev{-H}{\psi},
\end{align*}
where in the last step we used the fact that each entry of the Hamiltonian appears only in one term of the decomposition (see  \cref{remark:each_element_appears_once}).
\end{proof}

\begin{lemma}\label{lemma:verifier2_StoqSH_in_StoqMA}
    The probability that $\cV_2$ accepts the witness state $\ket{\psi} = \sum_{x} \alpha_x \ket{x}$ is    
    \begin{equation*}
        \Pr[\cV_2 \text{ accepts } \ket{\psi}] = 1 - \frac{1}{2d^2}\sum_x \alpha_x^2 |H_{x,f(x)}|.
    \end{equation*}
\end{lemma}
\begin{proof}
 Again, we are going to start by restricting ourselves to the case of a $1$-sparse term $H_j$. 
    
    The stoquastic verifier $\cV_2$  receives the witness $\ket{\psi}$ and prepares the state $$\ket{+}\ket{\psi}\ket{+}^{\otimes \ell} \ket{0}^{\otimes \ell} \ket{0}.$$
    Conditioned on the first qubit, we apply the same procedure as in \cref{lemma:verifier1_StoqSH_in_StoqMA} to construct $\ket{c_x}$ and we get the state
\begin{equation*}
    \ket{\Phi'} = \frac{1}{\sqrt{2}}\sum_x \alpha_x \ket{0} \ket{x} \ket{+}^{\otimes \ell} \ket{0}^{\otimes \ell} \ket{0} + \frac{1}{\sqrt{2}} \sum_x \alpha_x \ket{1}\ket{x}\ket{c_x}.
\end{equation*}

Now, we are going to measure the last qubit, so the measurement is with respect to the operator $\Pi_\acc = I_\text{else} \otimes \ketbra{0}_\text{final}$. We can observe that
    \begin{equation*}
        \Pi_\acc \ket{\Phi'} 
        = \frac{1}{\sqrt{2}}\sum_x \alpha_x \ket{0}\ket{x}\ket{+}^{\otimes \ell} \ket{0}^{\otimes \ell} \ket{0}
        + \frac{1}{\sqrt{2}} \sum_{x} \alpha_x \sqrt{1-|H^{(j)}_{x,f(x)}|}\ket{1}\ket{x}\ket{\chi_x'}\ket{0}.
    \end{equation*}
    So, the acceptance probability for the $H_j$ term will be
    \begin{align*}
        \Pr[\cV_2 \text{ accepts } \ket{\psi} \mid H_j] 
            &= \ev{\Pi_\text{acc}}{\Phi'} \\
            &= \frac{1}{2}\sum_x \alpha_x^2 + \frac{1}{2}\sum_{x} \alpha_x^2 (1- |H_{x,f(x)}|) \\
            &= \sum_x \alpha_x^2 - \frac{1}{2}\sum_x \alpha_x^2 |H^{(j)}_{x,f(x)}|.
    \end{align*}

    Notice that  $\sum_x \alpha^2_x = 1$ and that the second term of RHS contains only terms of $x\in H_j$, so the above probability can be rewritten as:
    \begin{equation*}
        \Pr[\cV_2 \text{ accepts } \ket{\psi} \mid H_j] =
            1 - \frac{1}{2} \sum_{x\in H_j} \alpha_x^2 |H_{x,f(x)}|.
    \end{equation*}
    
In the general case, the verifier $\cV_2$ chooses an integer $j$ from $1$ to $d^2$ uniformly at random. Thus, by averaging over all $j$, the total probability of $\cV_2$ accepting is 
\begin{align*}
    \Pr[\cV_2 \text{ accepts } \ket{\psi}] 
        &= \sum_{j=1}^{d^2} \frac{1}{d^2}\Pr[\cV_2 \text{ accepts} \ket{\psi} \mid H_j]\\
        &= \frac{1}{d^2} \sum_{j=1}^{d^2} \Big(1 - \frac{1}{2}\sum_{x\in H_j} \alpha_x^2 |H_{x,f(x)}|\Big)\\
        &= 1 - \frac{1}{2d^2} \sum_{x} \alpha_x^2 |H_{x,f(x)}|,
\end{align*}
    which concludes the proof.
\end{proof}
Now, we go on to put everything together as follows.

\begin{theorem}\label{thm:StoqSH_in_StoqMA}
    $\StoqSH_{\alpha, \beta} \in \StoqMA_{a, b}$, where $a = \frac{3}{4} + \frac{|\alpha|}{8d^2}$ and $b = \frac{3}{4} + \frac{|\beta|}{8d^2}$.
\end{theorem}

\begin{proof}
    First we will show that $\StoqSH_{\alpha, \beta} \in \gStoqMA_{a', b'}$ for some values of $a', b'$. The overall verifier $\cV$ accepts with probability 
    \begin{equation*}
        \Pr[\cV \text{ accepts } \ket{\psi}] = \frac{1}{2}\Pr[\cV_1 \text{ accepts } \ket{\psi}] +  \frac{1}{2}\Pr[\cV_2 \text{ accepts } \ket{\psi}].
    \end{equation*}
    From \cref{lemma:verifier1_StoqSH_in_StoqMA} and \cref{lemma:verifier2_StoqSH_in_StoqMA} we have that the total acceptance probability is 

    \begin{equation*}
        \Pr[\cV \text{ accepts } \ket{\psi}] = \frac{1}{2} -  \frac{1}{4d^2}\ev{H}{\psi}.
    \end{equation*}
    For completeness, the promise is that $\lambda_{\min{}}(H) \leq \alpha$, so there exists some state for which the verifier $\cV$ accepts with probability at least $a' = \frac{1}{2} + \frac{|\alpha|}{4d^2}$. Similarly, for soundness we have that  $\lambda_{\min{}}(H) \geq \beta$, so for all quantum states, the verifier $\cV$ accepts with probability at most $b' = \frac{1}{2} + \frac{|\beta|}{4d^2}$. It is easy to see that the gap between the completeness and soundness parameters remains at least $1/\poly(n)$. Finally, from \cref{thm:gStoqMA=StoqMA}, we have that $\StoqSH_{\alpha, \beta} \in \StoqMA_{a, b}$, where $a = \frac{1+a'}{2}$ and $b = \frac{1+b'}{2}$.
\end{proof}

\section{Separable Stoquastic Sparse Hamiltonians and StoqMA with multiple provers}

In this section, we consider the stoquastic version of the Separable Sparse Hamiltonian problem and we prove that is complete for $\StoqMA(2)$.
In \cref{sec:stoqma(k)}, we introduce a new complexity class called $\StoqMA(k)$, a version of $\StoqMA$ with $k$ unentangled proofs, and in \cref{sec:sepstoqsh} we present the problem of Separable Sparse Hamiltonians and the proof that this problem is $\StoqMA(2)$-complete. Then, in \cref{sec:stoqma(k)=stoqma(2)}, we show that in certain cases $k$ Merlins can be simulated by two Merlins in the stoquastic setting.

\subsection{StoqMA with unentangled proofs}\label{sec:stoqma(k)}

Let us first define the complexity class $\StoqMA(k)$.

\begin{definition}\label{def:StoqMA(k)} A promise problem $\cL = (\cL_\yes, \cL_\no)$ is in $\StoqMA_{a, b}(k)$ if there exists a uniform family of stoquastic verifiers, such that for any input $x\in\cL$ and for any fixed number of input bits $n$, there exist efficiently computable functions $a(n)$, $b(n)$ and $f_1, \dots, f_k$, and the corresponding verifier $\cV$ obeys:
    \begin{itemize}
        \item[] \emph{\textbf{Completeness}}: If $x \in \cL_\yes$, then there exist $k$ witnesses $\ket{\psi_1}, \dots, \ket{\psi_k}$, where each witness $\ket{\psi_i}$ consists of $f_i(n)$ qubits such that $$\Pr[\cV \text{ accepts } \ket{x}\otimes\ket{\psi_1}\otimes\dots\otimes\ket{\psi_k}] \geq a(n).$$
        \item[] \emph{\textbf{Soundness}}: If $x \in \cL_\no$, then for all $k$ witnesses $\ket{\psi_1}, \dots, \ket{\psi_k}$, where each witness $\ket{\psi_i}$ consists of $f_i(n)$ qubits, $$\Pr[\cV \text{ accepts } \ket{x}\otimes\ket{\psi_1}\otimes\dots\otimes\ket{\psi_k}] \leq b(n).$$
        \end{itemize}
\end{definition}

We can easily see that $\StoqMA(k)$ shares some properties with $\StoqMA$. First, one can easily see that there is a non-negative state that maximizes the acceptance probability. And secondly, one can see that the soundness is always greater or equal than $\frac{1}{2}$.

We can also define $\gStoqMA(k)$ as follows.
\begin{definition}[$\gStoqMA(k)$]\label{def:gStoqMA(k)}
    $\gStoqMA_{a,b}(k)$ is the same as $\StoqMA_{a,b}(k)$, except the verifier is a generalized stoquastic verifier for every input size $n$ with completeness parameter $a(n)$ and soundness parameter $b(n)$.
\end{definition}

\begin{corollary}\label{cor:gstoqMA(k)_in_stoqMA(2)}
Notice that \cref{thm:gStoqMA=StoqMA} continues to hold for $k$ Merlins as well, so we have that $\gStoqMA_{c,s}(k) \subseteq \StoqMA_{c', s'}(k)$, for $c' = \frac{1+c}{2}$ and $s' = \frac{1+s}{2}$.
\end{corollary}
\subsection{Separable Stoquastic Sparse Hamiltonians}\label{sec:sepstoqsh}

We now define the Separable Stoquastic Sparse Hamiltonian.

\begin{definition}(Separable Stoquastic Sparse Hamiltonian Problem $(\SepStoqSH_{\alpha, \beta})$) \label{def:SepStoqSH}
For a given Stoquastic Sparse Hamiltonian $H$ on $n$ qubits together with a partition of the qubits to disjoint sets $A$ and $B$, decide whether there exists a state $\ket{\psi} = \ket{\chi_A}\otimes \ket{\chi_B}$ such that $\ev{H}{\psi} \leq \alpha$ or if $\ev{H}{\psi} \geq \beta$ for all product states $\ket{\psi} = \ket{\chi_A} \otimes \ket{\chi_B}$, where $\beta - \alpha \geq 1/\poly(n)$.
\end{definition}

We then state the main result of this section.

\begin{theorem}\label{thm:SepStoqSH_is_StoqMA2_complete}
    $\SepStoqSH_{\alpha, \beta} \text{ is } \StoqMA_{c, s}(2)\text{-complete}$, where $c-s\geq 1/\poly(n)$.
\end{theorem}

To show inclusion, Arthur receives the states $\ket{\psi_1}$ and $\ket{\psi_2}$ from the two Merlins and runs the verification procedure of \cref{sec:protocol} with the witness $\ket{\psi} = \ket{\psi_1} \otimes \ket{\psi_2}$. So, from \cref{thm:StoqSH_in_StoqMA} we get the following.

\begin{corollary}\label{corr:SepStoqSH_in_StoqMA(2)}
    $\SepStoqSH_{\alpha, \beta} \in \StoqMA_{c, s}(2)$, where $c = \frac{3}{4} + \frac{|\alpha|}{8d^2}$ and $s= \frac{3}{4} + \frac{|\beta|}{8d^2}$.
\end{corollary}

So, it remains to show that $\SepStoqSH$ is $\StoqMA(2)$-hard by applying Kitaev's circuit-to-Hamiltonian construction following the proof of $\QMA(2)$-hardness of Separable Sparse Hamiltonians from \cite{CS12}. 

In their approach, the challenge is to make the history state separable. To this end, the verifier circuit has to have a specific structure. Suppose that both the first and the second witnesses contain $r$ registers. The first $r$ steps of the verification procedure consist of applying swap tests between the corresponding registers of the witnesses, and then Arthur continues the verification procedure using only the first witness. Notice that the SWAP operators applied to the circuit are what make the Hamiltonian obtained a sparse Hamiltonian.

To show completeness, the two Merlins send identical states, so the history state is separable and using arguments from \cite{kitaevbook}, the energy of the Hamiltonian will be low. For soundness, \cite{CS12} assume that there exists a low-energy product state and they show that the history state that corresponds to that product state has high energy, leading to a contradiction.

The problem in our case is, however, that we cannot perform error amplification in the completeness and soundness parameters for $\StoqMA(2)$, which implies that we need to use the perturbation approach of \cite{BBT06}. The tricky part here is, thus, to choose the correct parameters, so that we have an inverse polynomial energy gap between the yes and no cases. We end up with the following theorem.

\begin{theorem}\label{thm:SepStoqSH_is_StoqMA(2)_hard}
    $\SepStoqSH_{\alpha, \beta}$ is $\StoqMA_{c,s}(2)$-hard.
\end{theorem}

The proof is deferred to \cref{app:SepStoqSH}.
\subsection{StoqMA(k) vs StoqMA(2)}\label{sec:stoqma(k)=stoqma(2)}

We will show that, for certain values of the completeness and soundness parameters ,$k$ Merlins can be simulated by only two Merlins, i.e. $\StoqMA(k) \subseteq \StoqMA(2)$, for $k \geq 2$. Since it is not known how to perform error reduction for $\StoqMA(k)$, our proof does not generalize for all values of $c, s$.

\begin{theorem}\label{thm:stoqma(k)=stoqma(2)}
    $\StoqMA_{c,s}(k) \subseteq \StoqMA_{c', s'}(2)$, for values of $c,s$ such that $c' - s' \geq 1/\poly(n)$, 
    \begin{equation*}
        c' = \frac{3+c}{4} \text{ and }
        s' = \frac{11}{12} + \frac{1}{12}(s^2-2s)^2.
    \end{equation*}
\end{theorem}

The proof is based on the proof of $\QMA(k) = \QMA(2)$, which was proven in \cite{HM10} using the Product Test presented below.
\begin{procedure}[Product Test \cite{HM10}]\label{proc:product-test} The product test works as follows.
    \begin{enumerate}
        \item Receive two states $\ket{\psi_1}, \ket{\psi_2} \in \C^{d_1}\otimes  \cdots \otimes \C^{d_k}$.
        \item Perform a swap test on each of the $k$ subsystems
        \item If all the tests accepted, return ``\textsf{Accept}'', else ``\textsf{Reject}''.
    \end{enumerate}
\end{procedure}

Now we present the $\StoqMA(2)$ protocol that simulates any $\StoqMA(k)$ problem for certain values of the completeness and soundness parameters, which is based on the $\QMA(k)$-to-$\QMA(2)$ protocol of \cite{HM10}. Let $\mathcal{A}$ denote the verification procedure of the original $\StoqMA(k)$ protocol.

\begin{procedure}[$\StoqMA(k)$ to $\StoqMA(2)$]\label{proc:stoqma(k)-to-stoqma(2)} The $\StoqMA(2)$ protocol proceeds as follows.
\begin{enumerate}
    \item Both Merlins send $\ket{\psi} = \ket{\psi_1}\otimes\dots \otimes \ket{\psi_k}$ to Arthur.
    \item Arthur performs one of the following tests with probability 1/2, 
    \begin{enumerate}
        \item Run the product test with the two states as input.  
        \item Choose randomly one of the states from the two Merlins and run the original verification procedure $\cA$ on that state.
    \end{enumerate}
\end{enumerate}
\end{procedure}

Technically, we can only run the product test in $\gStoqMA(2)$, so we first show the following.

\begin{theorem}\label{thm:gstoqma(k)=gstoqma(2)}
    $\gStoqMA_{c,s}(k) \subseteq \gStoqMA_{c', s'}(2)$ for values of $c,s$ that  satisfy $c' - s' \geq 1/\poly(n)$, 
    \begin{equation*}
        c' = \frac{1+c}{2} \text{ and }
        s' = \frac{5}{6} + \frac{1}{6}(s^2-2s)^2.
    \end{equation*}
\end{theorem}
The proof of \cref{thm:gstoqma(k)=gstoqma(2)} follows the proof of $\QMA(k) \subseteq \QMA(2)$ of \cite{HM10} with the improved analysis of the product test by \cite{SW22} and is deferred to \cref{app:gstoqma(k)=gstoqma(2)}.

\begin{proof}[Proof of \cref{thm:stoqma(k)=stoqma(2)}]
We have that
    \begin{equation*}
        \StoqMA_{c,s}(k) \subseteq \gStoqMA_{c,s}(k) \subseteq \gStoqMA_{c',s'}(2) \subseteq \StoqMA_{c'', s''}(2),
    \end{equation*}
    where the second inclusion comes from \cref{thm:gstoqma(k)=gstoqma(2)}, the third inclusion holds from \cref{cor:gstoqMA(k)_in_stoqMA(2)}, and carrying out the completeness and soundness parameters we have that $c'' = \frac{3+c}{4}$ and $s'' = \frac{11}{12} + \frac{1}{12}(s^2-2s)^2$, which completes the proof.
\end{proof}

\section*{Acknowledgments}
MR recognizes partial support by project MIS 5154714 of the National Recovery and Resilience Plan Greece 2.0 funded by the European Union under the NextGenerationEU Program. ABG and MR are supported by ANR JCJC TCS-NISQ ANR-22-CE47-0004.

\bibliographystyle{alpha}   
\bibliography{references}   

@inproceedings{BCCKS14,
author = {Berry, Dominic W. and Childs, Andrew M. and Cleve, Richard and Kothari, Robin and Somma, Rolando D.},
title = {Exponential improvement in precision for simulating sparse Hamiltonians},
year = {2014},
isbn = {9781450327107},
publisher = {Association for Computing Machinery},
address = {New York, NY, USA},
url = {https://doi.org/10.1145/2591796.2591854},
doi = {10.1145/2591796.2591854},
booktitle = {Proceedings of the Forty-Sixth Annual ACM Symposium on Theory of Computing},
pages = {283–292},
numpages = {10},
location = {New York, New York},
series = {STOC '14}
}

@article{BDOT08,
author = {Bravyi, Sergey and Divincenzo, David P. and Oliveira, Roberto and Terhal, Barbara M.},
title = {The complexity of stoquastic local Hamiltonian problems},
year = {2008},
issue_date = {May 2008},
publisher = {Rinton Press, Incorporated},
address = {Paramus, NJ},
volume = {8},
number = {5},
issn = {1533-7146},
journal = {Quantum Info. Comput.},
month = may,
pages = {361–385},
numpages = {25}
}

@misc{BBT06,
      title={Merlin-{A}rthur {G}ames and {S}toquastic {C}omplexity}, 
      author={Sergey Bravyi and Arvid J. Bessen and Barbara M. Terhal},
      year={2006},
      eprint={quant-ph/0611021},
      archivePrefix={arXiv},
      primaryClass={quant-ph},
      url={https://arxiv.org/abs/quant-ph/0611021}, 
}

@inproceedings{ATS03,
author = {Aharonov, Dorit and Ta-Shma, Amnon},
title = {Adiabatic quantum state generation and statistical zero knowledge},
year = {2003},
isbn = {1581136749},
publisher = {Association for Computing Machinery},
address = {New York, NY, USA},
url = {https://doi.org/10.1145/780542.780546},
doi = {10.1145/780542.780546},
booktitle = {Proceedings of the Thirty-Fifth Annual ACM Symposium on Theory of Computing},
pages = {20–29},
numpages = {10},
location = {San Diego, CA, USA},
series = {STOC '03}
}

@book{NC10, place={Cambridge}, title={Quantum Computation and Quantum Information: 10th Anniversary Edition}, publisher={Cambridge University Press}, author={Nielsen, Michael A. and Chuang, Isaac L.}, year={2010}}

@misc{Ioannou,
    author = "Ioannou, Marios and Piddock, Stephen and Marvian, Milad and Klassen, Joel and Terhal, Barbara M.",
    title = "{Termwise versus globally stoquastic local Hamiltonians: questions of complexity and sign-curing}",
    eprint = "2007.11964",
    archivePrefix = "arXiv",
    primaryClass = "quant-ph",
    month = "7",
    year = "2020"
}

@INPROCEEDINGS{CS12,
  author={Chailloux, André and Sattath, Or},
  booktitle={2012 IEEE 27th Conference on Computational Complexity}, 
  title={The Complexity of the Separable Hamiltonian Problem}, 
  year={2012},
  volume={},
  number={},
  pages={32-41},
  doi={10.1109/CCC.2012.42}}

@book{kitaevbook,
author = {Kitaev, A. Yu. and Shen, A. H. and Vyalyi, M. N.},
title = {Classical and Quantum Computation},
year = {2002},
isbn = {0821832298},
publisher = {American Mathematical Society},
address = {USA}
}

@inproceedings{SW22,
author = {Mehdi Soleimanifar and John Wright},
title = {Testing matrix product states},
booktitle = {Proceedings of the 2022 Annual ACM-SIAM Symposium on Discrete Algorithms (SODA)},
chapter = {},
year={2022},
pages = {1679-1701},
doi = {10.1137/1.9781611977073.68},
URL = {https://epubs.siam.org/doi/abs/10.1137/1.9781611977073.68},
eprint = {https://epubs.siam.org/doi/pdf/10.1137/1.9781611977073.68}
}

@article{HM10,
author = {Harrow, Aram W. and Montanaro, Ashley},
title = {Testing Product States, Quantum Merlin-Arthur Games and Tensor Optimization},
year = {2013},
issue_date = {February 2013},
publisher = {Association for Computing Machinery},
address = {New York, NY, USA},
volume = {60},
number = {1},
issn = {0004-5411},
url = {https://doi.org/10.1145/2432622.2432625},
doi = {10.1145/2432622.2432625},
journal = {J. ACM},
month = feb,
articleno = {3},
numpages = {43}
}

@misc{stoquasticconsistencyliu,
      title={The Local Consistency Problem for Stoquastic and 1-D Quantum Systems}, 
      author={Yi-Kai Liu},
      year={2007},
      eprint={0712.1388},
      archivePrefix={arXiv},
      primaryClass={quant-ph},
      url={https://arxiv.org/abs/0712.1388}, 
}

@inproceedings{am2,
author = {Aaronson, Scott and Impagliazzo, Russell and Moshkovitz, Dana},
title = {AM with Multiple Merlins},
year = {2014},
isbn = {9781479936267},
publisher = {IEEE Computer Society},
address = {USA},
url = {https://doi.org/10.1109/CCC.2014.13},
doi = {10.1109/CCC.2014.13},
booktitle = {Proceedings of the 2014 IEEE 29th Conference on Computational Complexity},
pages = {44–55},
numpages = {12},
series = {CCC '14}
}

@article{Hastings2021powerofadiabatic,
  doi = {10.22331/q-2021-12-06-597},
  url = {https://doi.org/10.22331/q-2021-12-06-597},
  title = {The {P}ower of {A}diabatic {Q}uantum {C}omputation with {N}o {S}ign {P}roblem},
  author = {Hastings, Matthew B.},
  journal = {{Quantum}},
  issn = {2521-327X},
  publisher = {{Verein zur F{\"{o}}rderung des Open Access Publizierens in den Quantenwissenschaften}},
  volume = {5},
  pages = {597},
  month = dec,
  year = {2021}
}

@inproceedings{gilyenHastingsVazirani,
author = {Gily\'{e}n, Andr\'{a}s and Hastings, Matthew B. and Vazirani, Umesh},
title = {(Sub)Exponential advantage of adiabatic Quantum computation with no sign problem},
year = {2021},
isbn = {9781450380539},
publisher = {Association for Computing Machinery},
address = {New York, NY, USA},
url = {https://doi.org/10.1145/3406325.3451060},
doi = {10.1145/3406325.3451060},
booktitle = {Proceedings of the 53rd Annual ACM SIGACT Symposium on Theory of Computing},
pages = {1357–1369},
numpages = {13},
location = {Virtual, Italy},
series = {STOC 2021}
}

@article{FixedNodeHamiltonian,
  title = {Proof for an upper bound in fixed-node Monte Carlo for lattice fermions},
  author = {ten Haaf, D. F. B. and van Bemmel, H. J. M. and van Leeuwen, J. M. J. and van Saarloos, W. and Ceperley, D. M.},
  journal = {Phys. Rev. B},
  volume = {51},
  issue = {19},
  pages = {13039--13045},
  numpages = {0},
  year = {1995},
  month = {May},
  publisher = {American Physical Society},
  doi = {10.1103/PhysRevB.51.13039},
  url = {https://link.aps.org/doi/10.1103/PhysRevB.51.13039}
}

@article{Bravyi2023rapidlymixingmarkov,
  doi = {10.22331/q-2023-11-07-1173},
  url = {https://doi.org/10.22331/q-2023-11-07-1173},
  title = {A rapidly mixing {M}arkov chain from any gapped quantum many-body system},
  author = {Bravyi, Sergey and Carleo, Giuseppe and Gosset, David and Liu, Yinchen},
  journal = {{Quantum}},
  issn = {2521-327X},
  publisher = {{Verein zur F{\"{o}}rderung des Open Access Publizierens in den Quantenwissenschaften}},
  volume = {7},
  pages = {1173},
  month = nov,
  year = {2023}
}

@article{JiangSuccinct,
  title = {Local Hamiltonian Problem with Succinct Ground State is MA-Complete},
  author = {Jiang, Jiaqing},
  journal = {PRX Quantum},
  volume = {6},
  issue = {2},
  pages = {020312},
  numpages = {28},
  year = {2025},
  month = {Apr},
  publisher = {American Physical Society},
  doi = {10.1103/PRXQuantum.6.020312},
  url = {https://link.aps.org/doi/10.1103/PRXQuantum.6.020312}
}

@article{classiciation2localHamiltonians,
author = {Cubitt, Toby and Montanaro, Ashley},
title = {Complexity Classification of Local Hamiltonian Problems},
journal = {SIAM Journal on Computing},
volume = {45},
number = {2},
pages = {268-316},
year = {2016},
doi = {10.1137/140998287},

URL = { 
    
        https://doi.org/10.1137/140998287
    
    

},
eprint = { 
    
        https://doi.org/10.1137/140998287
    
    

}
}

@article{BT10,
author = {Bravyi, Sergey and Terhal, Barbara},
title = {Complexity of Stoquastic Frustration-Free Hamiltonians},
journal = {SIAM Journal on Computing},
volume = {39},
number = {4},
pages = {1462-1485},
year = {2010},
doi = {10.1137/08072689X},

URL = { 
    
        https://doi.org/10.1137/08072689X
    
    

},
eprint = { 
    
        https://doi.org/10.1137/08072689X
    
    

}
}

@inproceedings{AG19,
  title={Stoquastic PCP vs. Randomness},
  author={Aharonov, Dorit and Grilo, Alex Bredariol},
  booktitle={2019 IEEE 60th Annual Symposium on Foundations of Computer Science (FOCS)},
  pages={1000--1023},
  year={2019},
  organization={IEEE},
doi = {10.1109/FOCS.2019.00065}
}

@misc{raza2025complexitygeometricallylocalstoquastic,
      title={Complexity of geometrically local stoquastic Hamiltonians}, 
      author={Asad Raza and Jens Eisert and Alex B. Grilo},
      year={2025},
      eprint={2407.15499},
      archivePrefix={arXiv},
      primaryClass={quant-ph},
      url={https://arxiv.org/abs/2407.15499}, 
}

@misc{waite2025complexitylocalstoquastichamiltonians,
      title={The Complexity of Local Stoquastic Hamiltonians on 2D Lattices}, 
      author={Gabriel Waite and Michael J. Bremner},
      year={2025},
      eprint={2502.14244},
      archivePrefix={arXiv},
      primaryClass={quant-ph},
      url={https://arxiv.org/abs/2502.14244}, 
}

@article{AGL25,
  doi = {10.22331/q-2025-09-11-1853},
  url = {https://doi.org/10.22331/q-2025-09-11-1853},
  title = {Stoq{MA} vs. {MA}: the power of error reduction},
  author = {Aharonov, Dorit and Grilo, Alex B. and Liu, Yupan},
  journal = {{Quantum}},
  issn = {2521-327X},
  publisher = {{Verein zur F{\"{o}}rderung des Open Access Publizierens in den Quantenwissenschaften}},
  volume = {9},
  pages = {1853},
  month = sep,
  year = {2025}
}

@article{BH17,
  title={On complexity of the quantum Ising model},
  author={Bravyi, Sergey and Hastings, Matthew},
  journal={Communications in Mathematical Physics},
  volume={349},
  number={1},
  pages={1--45},
  year={2017},
  publisher={Springer},
doi = {10.1007/s00220-016-2787-4}
}

@article{Kit99,
  title={Quantum $\mathsf{NP}$},
  author={Kitaev, Alexei},
  journal={Talk at AQIP},
  volume={99},
  year={1999}
}

@article{piddockmontanaro,
author = {Piddock, Stephen and Montanaro, Ashley},
title = {The complexity of antiferromagnetic interactions and 2D lattices},
year = {2017},
issue_date = {June 2017},
publisher = {Rinton Press, Incorporated},
address = {Paramus, NJ},
volume = {17},
number = {7–8},
issn = {1533-7146},
journal = {Quantum Info. Comput.},
month = jun,
pages = {636–672},
numpages = {37}
}

@article{swaptest,
  title = {Quantum Fingerprinting},
  author = {Buhrman, Harry and Cleve, Richard and Watrous, John and de Wolf, Ronald},
  journal = {Phys. Rev. Lett.},
  volume = {87},
  issue = {16},
  pages = {167902},
  numpages = {4},
  year = {2001},
  month = {Sep},
  publisher = {American Physical Society},
  doi = {10.1103/PhysRevLett.87.167902},
  url = {https://link.aps.org/doi/10.1103/PhysRevLett.87.167902}
}

@article{2localkempe,
author = {Kempe, Julia and Kitaev, Alexei and Regev, Oded},
title = {The Complexity of the Local Hamiltonian Problem},
journal = {SIAM Journal on Computing},
volume = {35},
number = {5},
pages = {1070-1097},
year = {2006},
doi = {10.1137/S0097539704445226},

URL = { 
    
        https://doi.org/10.1137/S0097539704445226
    
    

},
eprint = { 
    
        https://doi.org/10.1137/S0097539704445226
    
    

}
}

@InProceedings{hubbardchilds,
author="Childs, Andrew M.
and Gosset, David
and Webb, Zak",
editor="Esparza, Javier
and Fraigniaud, Pierre
and Husfeldt, Thore
and Koutsoupias, Elias",
title="The Bose-Hubbard Model is QMA-complete",
booktitle="Automata, Languages, and Programming",
year="2014",
publisher="Springer Berlin Heidelberg",
address="Berlin, Heidelberg",
pages="308--319",
isbn="978-3-662-43948-7"
}

@INPROCEEDINGS{1dlineaharonov,
  author={Aharonov, Dorit and Gottesman, Daniel and Irani, Sandy and Kempe, Julia},
  booktitle={48th Annual IEEE Symposium on Foundations of Computer Science (FOCS'07)}, 
  title={The Power of Quantum Systems on a Line}, 
  year={2007},
  volume={},
  number={},
  pages={373-383},
  doi={10.1109/FOCS.2007.46}}

@misc{liuUnentangledStoquasticMerlinArthur2026,
  title = {Unentangled Stoquastic {{Merlin-Arthur}} Proof Systems: The Power of Unentanglement without Destructive Interference},
  shorttitle = {Unentangled Stoquastic {{Merlin-Arthur}} Proof Systems},
  author = {Liu, Yupan and Wu, Pei},
  year = 2026,
  month = apr,
  number = {arXiv:2604.27886},
  eprint = {2604.27886},
  primaryclass = {quant-ph},
  publisher = {arXiv},
  doi = {10.48550/arXiv.2604.27886},
  urldate = {2026-05-04},
  archiveprefix = {arXiv}
}

\appendix

\newpage
\section{Proof of \cref{thm:SepStoqSH_is_StoqMA(2)_hard} }\label{app:SepStoqSH}
\begin{proof}
As discussed, we follow the proof of \cite{CS12}.

    Consider a promise problem $\cL \in \StoqMA_{c,s}(2)$ and let $\cV = (n, n_w, n_0, n_+, U)$ be the corresponding verifier, where the circuit $U$ can be decomposed into $T$ unitaries $U= U_T\dots U_1 U_0$, where $U_0 = I$. We will assume that the verifier has a specific structure: the verifier initially receives the two proofs, runs SWAP tests between them and then performs the verification algorithm on the first proof. This can be done, without loss of generality, since generalized stoquastic verifiers can perform multi-qubit measurements and this version of verifier can simulate any other verifier of $\StoqMA(2)$ using the \cref{proc:stoqma(k)-to-stoqma(2)} and \cref{thm:gstoqma(k)=gstoqma(2)}.

    We will apply Kitaev's construction \cite{kitaevbook} of the circuit with the modification that we will add the output Hamiltonian as a small perturbation, as in \cite{BBT06}, so we will consider a Hamiltonian of the form:
    \begin{equation*}
        \widetilde{H} = H_\init + H_\prop +  \delta H_\out, \quad 0 < \delta \ll \Delta
    \end{equation*}
    where we follow the same definition for the Hamiltonians as in Kitaev's proof, and $\Delta$ is the spectral gap of $H_\init + H_\prop$, which is quantified below.
    \begin{lemma}(Claim 16 of \cite{CS12})
        The smallest eigenvalue of $H_\init + H_\prop$ is 0 and the second smallest eigenvalue is $\Delta = \frac{C}{(T+1)^3}$, where $C$ is some universal constant.
    \end{lemma}
    To be able to use the perturbative approach, we will need to ensure that $\delta \ll \Delta$. For this reason, let us fix the following parameters:
    \begin{equation*}
        \alpha := \frac{(c-s)^2}{1024(T+1)^2}, \quad
        \delta := \frac{1}{\big(1-c+\frac{c-s}{4}\big)}\cdot \frac{\alpha C}{(T+1)^2} 
    \end{equation*}
    We are going to show the following:
    \begin{itemize}
        \item If $x\in \cL_\yes$, then there exists $\ket{\psi} = \ket{\psi_1}\otimes\ket{\psi_2}$, such that $\ev{\widetilde{H}}{\psi} \leq \frac{\delta}{T+1}(1-c)$.
        \item If $x\in \cL_\no$, then for all $\ket{\psi} = \ket{\psi_1}\otimes\ket{\psi_2}$, $\ev{\widetilde{H}}{\psi} \geq (1 + \frac{3}{4}\cdot \frac{c-s}{1-c}) \frac{\delta}{T+1}(1-c)$.
    \end{itemize}
    \paragraph{Completeness:} The two Merlins can send identical states $\ket{\psi_1} = \ket{\psi_2}$ and Arthur initially applies swap tests on the registers of Merlins. The history states that correspond to the witness state $\ket{\psi}$ are going to be of the form
    \begin{equation*}
        \ket*{\eta_\psi} = \frac{1}{\sqrt{T+1}}\sum_{t=0}^T \ket{t}_C \otimes U_tU_{t-1}\dots U_0 \big(\ket{0^m}_A \otimes \ket{\psi_1}_{P_1} \otimes \ket{\psi_2}_{P_2}\big),
    \end{equation*}
    where $C$ is the clock subsystem, $A$ is the ancilla subsystem, and $P_1$ and $P_2$ are the first and second proof subsystems. In \cite{CS12}, it is shown that $\ket{\eta_\psi}$ is a tensor product state between the spaces $(C \otimes A \otimes P_1)$ and $P_2$. Kitaev's proof \cite{kitaevbook} shows that
    \begin{equation*}
        \ev*{(H_\init + H_\prop)}{\eta_\psi} = 0
    \end{equation*}
    and
    \begin{equation*}
        \ev*{H_\out}{\eta_\psi} = \frac{1}{T+1}(1-\Pr[\cV \text{ accepts } \ket{\psi}]) = \frac{1-c}{T+1}.
    \end{equation*}
    So, in total we get
    \begin{equation*}
        \ev*{\widetilde{H}}{\eta_\psi} \leq \frac{\delta}{T+1}(1-c).
    \end{equation*}

    \paragraph{Soundness:} We follow the approach used in \cite{CS12}. An overview of the proof is as  follows:
    \begin{itemize}
        \item Assume that there exists a low-energy product state $\ket{\omega} = \ket{\omega_1}\otimes \ket{\omega_2}$, then $\ket{\omega}$ is close to a history state $\ket{\eta_\psi}$.
        \item If $\ket{\eta_\psi}$ is close to a product state, then the corresponding witness $\ket{\psi}$ is close to a product state.
        \item If $\ket{\psi}$ is close to a product state, then $\ket{\eta_\psi}$ will have high energy, which leads to contradiction.
    \end{itemize}
    We are now going to state some important lemmas that will be used in the proof.
    \begin{lemma}[Lemma 15 of \cite{CS12}]\label{lemma:lemma15} If $\ev{\widetilde{H}}{\omega} \leq \alpha\frac{C}{(T+1)^3}$, then there exists a history state $\ket{\eta}$ such that $|\braket{\omega}{\eta}|^2 \geq 1 - \alpha$. \footnote{We notice that the original proof actually holds for the Hamiltonian $H_\init + H_\prop + H_\out$. However, in their proof, the result does not depend on $H_\out$, so our $\delta$ factor does not affect its results.} 
    \end{lemma}
    
    \begin{lemma}[Lemma 17 of \cite{CS12}]\label{lemma:lemma17} Let $\ket{\psi}$ a witness state and $\ket{\eta_\psi}$ a corresponding history state. If there exist two states $\ket{\psi_1} \in C \otimes A \otimes P_1$ and $\ket{\psi_2} \in P_2$ such that $|\bra{\eta_\psi}(\ket{\psi_1}\otimes \ket{\psi_2})|^2 \geq 1-\varepsilon$, then there exists a product state $\ket{L} \in P_1$ such that $$|\bra{\psi}(\ket{L}\otimes\ket{\psi_2})|^2 \geq 1 - \varepsilon(T+1).$$
    \end{lemma}
    \begin{lemma}[based on Lemma 18 of \cite{CS12}] \label{lemma:lemma18}Consider a history state $\ket{\eta_\psi}$ with the corresponding witness $\ket{\psi}$. If the soundness parameter is $s$, in the no instance, if $|\bra{\psi}(\ket{\psi_1}\otimes\ket{\psi_2})|^2 \geq 1- \varepsilon$, then $\ev{\widetilde{H}}{\eta_\psi} \geq \frac{\delta}{T+1}(1-s-2\sqrt{\varepsilon})$.
    \end{lemma}
    \begin{proof}[Proof of \cref{lemma:lemma18}]
    For every state $\ket{\psi}$, $(H_\init + H_\prop)\ket{\eta_\psi} = 0$, so $\ev{\widetilde{H}}{\eta_\psi} = \delta \ev{H_\out}{\eta_\psi}$.   
    Like in the original proof, we have
    \begin{equation*}
        \ev{H_\out}{\eta_\psi} = \frac{1}{T+1} \bra{0^m}\otimes\bra{\psi} U_1^\dagger \dots U_T^\dagger \Pi_\text{reject} U_T \dots U_1 \ket{0^m}\otimes \ket{\psi},
    \end{equation*}
    so 
    \begin{equation*}
        \ev{\widetilde{H}}{\eta_\psi} = \frac{\delta}{T+1} \cdot \Pr[\cV \text{ rejects } \ket{\psi}].
    \end{equation*}

    From \cref{lemma:close-overlap}, since $|\bra{\psi}(\ket{\psi_1}\otimes\ket{\psi_2})|^2 \geq 1- \varepsilon$, we have that
    \begin{equation*}
        \Pr[\cV \text{ rejects } \ket{\psi}] \geq \Pr[\cV \text{ rejects } \ket{\psi_1}\otimes\ket{\psi_2}] - 2\sqrt{\varepsilon} \geq 1-s-2\sqrt{\varepsilon}.
    \end{equation*}
    Substituting back, we get:
    \begin{equation*}
        \ev{\widetilde{H}}{\eta_\psi} \geq \frac{\delta}{T+1}(1-s-2\sqrt{\varepsilon}). \qedhere
    \end{equation*}
    \end{proof}

    \medspace 
    We will now combine the lemmas above to prove soundness. Assume that there exists a state $\ket{\omega} = \ket{\omega_1}\otimes\ket{\omega_2}$, with energy below the promise, i.e., assume that
    \begin{equation*}
        \ev{\widetilde{H}}{\omega} \leq
        \bigg(1+\frac{1}{4}\cdot\frac{c-s}{1-c}\bigg) \frac{\delta}{T+1}(1-c).
    \end{equation*}
    From the definition of $\delta$, we get that
    \begin{equation*}
        \ev{\widetilde{H}}{\omega} \leq \frac{\alpha C}{(T+1)^3},
    \end{equation*}
    so from \cref{lemma:lemma15}, there exists a state $\ket{\eta_\psi}$ such that $ |\braket*{\eta_\psi}{\omega}|^2 \geq 1 - \alpha$. From \cref{lemma:lemma17}, there exists some state $\ket{\phi} = \ket{\phi_1} \otimes \ket{\phi_2}$ such that $$|\braket{\psi}{\phi}|^2 \geq 1 - \alpha (T+1),$$ and, thus, from \cref{lemma:lemma18} with $\varepsilon = \alpha (T+1)$, we get 
    \begin{equation*}
        \ev{\widetilde{H}}{\eta_\psi} \geq \frac{\delta}{T+1}\Big[1-s-2\sqrt{\alpha(T+1)}\Big].
    \end{equation*}

    We will now find a lower bound of $\ev{\widetilde{H}}{\omega}$ using the overlap with $\ket{\eta_\psi}$. Indeed, we have that 
    \begin{equation*}
        \ket{\omega} = \sqrt{1-p}\ket{\eta_\psi} + \sqrt{p}\ket*{\eta_\psi^\perp},
    \end{equation*}
    for some $\ket*{\eta_\psi^\perp}$ such that $\braket*{\eta_\psi}{\eta_\psi^\perp} = 0$ and $0 \leq p \leq \alpha$. Then,
    \begin{equation*}
        \ev{\widetilde{H}}{\omega} = (1-p) \ev{\widetilde{H}}{\eta_\psi} + p\ev*{\widetilde{H}}{\eta_\psi^\perp} + 2\sqrt{(1-p)p}\, \text{Re}\big\{\mel*{\eta_\psi}{\widetilde{H}}{\eta_\psi^\perp}\big\}.
    \end{equation*}
    Since $\widetilde H\succeq 0$, the second term is non-negative. 
    Moreover, as $(H_{\init}+H_{\prop})\ket{\eta_\psi}=0$ and $0 \preceq H_{\out} \preceq I$, we have
    \begin{equation*} 
        \Re\{\mel*{\eta_\psi}{\widetilde H}{\eta_\psi^\perp}\}
        =\delta\,\Re\{\mel*{\eta_\psi}{H_{\out}}{\eta_\psi^\perp}\}\ge -\delta.
    \end{equation*}
    Using that $1-p \geq 1-\alpha$ and $\sqrt{(1-p)p} \leq \sqrt{p} \leq \sqrt{\alpha}$, we obtain
    \begin{equation*}
        \ev{\widetilde{H}}{\omega} \geq (1-\alpha) \ev{\widetilde{H}}{\eta_\psi} - 2\delta \sqrt{\alpha}.
    \end{equation*}
    Substituting the bound from \cref{lemma:lemma18}, we conclude that
    \begin{equation*}
        \ev{\widetilde{H}}{\omega} \geq \frac{\delta}{T+1}\Big[(1-\alpha)\Big(1-s-2\sqrt{\alpha(T+1)}\Big) - 2\sqrt{\alpha}(T+1)\Big]
    \end{equation*}
    From the definition of $\alpha$, we get
    \begin{equation*}
        2\sqrt{\alpha (T+1)} = \frac{c-s}{16\sqrt{T+1}} \leq \frac{c-s}{16}, \quad 2\sqrt{\alpha}(T+1) = \frac{c-s}{16}.
    \end{equation*}
    Using the fact that $1-s = (1-c) + (c-s)$, we expand:
    \begin{equation*}
        \ev{\widetilde{H}}{\omega} \geq \frac{\delta}{T+1}\Big[(1-c) - \alpha(1-c) + (c-s) \Big(1-\alpha  - \frac{1}{8}\Big)\Big].
    \end{equation*}
    Since $\alpha(1-c) \leq \frac{(c-s)}{1024}$ and $\alpha \leq \frac{1}{1024}$, a straightforward calculation yields
    \begin{equation*}
        \ev{\widetilde{H}}{\omega} \geq \frac{\delta}{T+1}\Big[(1-c) + \frac{3}{4}(c-s)\Big] 
        = \bigg(1 + \frac{3}{4}\cdot \frac{c-s}{1-c}\bigg) \frac{\delta}{T+1}(1-c),
    \end{equation*}
contradicting the initial assumption. This completes the proof of the soundness property. 
\end{proof}

\newpage
\section{Proof of \cref{thm:gstoqma(k)=gstoqma(2)}}\label[app]{app:gstoqma(k)=gstoqma(2)}

We first introduce some notation. Assume that Arthur receives two states $\ket{\psi_1}, \ket{\psi_2} \in \C^{d_1} \otimes \cdots \otimes \C^{d_k}$ on $k$ registers. Let $\rho := \ketbra{\psi_1}$ and $\sigma : = \ketbra{\psi_2}$.  For some subset of registers $S\subseteq [k]$, we denote $\rho_S := \Tr_{\overline{S}} (\rho)$, where $\overline{S} = [k]\, \backslash\, S$.

Let $\SWAP_i$ be the $\SWAP$ operator that acts on the $i$-th register of the two copies and, for any subset $S \subseteq [n]$, we denote the swap operator that acts on $S$ as $\SWAP_S:= \Pi_{i\in S} \SWAP_i$. 

Let $\PT_k(\ket{\psi_1}, \ket{\psi_2})$ (resp. $\PT_k(\rho, \sigma)$) be the acceptance probability of the product test on $\ket{\psi_1}$ and $\ket{\psi_2}$ (resp. on $\rho$ and $\sigma$), and we denote $\PT_k(\ket{\psi}) := \PT_k(\ket{\psi}, \ket{\psi})$ (resp. $\PT(\rho):=\PT_k(\rho,\rho)$). 

\medspace 

The improved analysis of the product test of \cite{SW22} achieves the following.
\begin{theorem}[Product test upper bound from \cite{SW22}]\label{thm:prod-test-sw-upper-bound} 
Let $\ket{\psi} \in \C^{d_1}\otimes \dots \otimes \C^{d_k}$ and  $1-\varepsilon = \max\{|\braket{\psi}{\phi}|^2 : \ket{\phi} \in \C^{d_1}\otimes \dots \otimes \C^{d_k}\}$.
For $k\geq 1$, 
\begin{equation*}
    \PT_k(\ket{\psi}) \leq 
    \begin{cases}
        1-\varepsilon+\varepsilon^2 & \text{ if } \varepsilon \leq \frac{1}{2},\\
        1 - \frac{2}{3}\varepsilon + \frac{1}{3}\varepsilon^2 &\text{otherwise}.
    \end{cases}
\end{equation*}
\end{theorem}

We then provide a useful lemma, based on \cite{HM10}, and its proof for the sake of completeness.

\begin{lemma}[Lemma 2 and Lemma 5 of \cite{HM10} (rephrased)] \label{lemma:prob-acc-pt}
    Let $\rho, \sigma$ two mixed states on $k$ registers. Then, the following holds for the probability of acceptance of the product test: 
    \begin{equation}
        \PT_k(\rho, \sigma) \leq \frac{\PT_k(\rho) + \PT_k(\sigma)}{2}
    \end{equation}
\end{lemma}
\begin{proof}
    For the product test, the acceptance POVM is 
    \begin{equation*}
        \Pi_{\PT} = \bigotimes_{i=1}^k \frac{I + \SWAP_i}{2}.
    \end{equation*}
    Expanding the tensor product this can be rewritten as
    \begin{equation*}
        \Pi_{\PT} = \frac{1}{2^k} \sum_{S \subseteq [k]} \SWAP_S.
    \end{equation*}
    Hence,
    \begin{equation*}
        \PT_k(\rho, \sigma) = \Tr(\Pi_{\PT}\, \rho \otimes \sigma) = \frac{1}{2^k} \sum_{S \subseteq [k]} \Tr(\rho_S \sigma_S).
    \end{equation*}
    Finally, observe that
    \begin{align*}
        \PT_k(\rho,\sigma) = \frac{1}{2^k} \sum_{S \subseteq [k]} \Tr(\rho_S \sigma_S) 
        &\leq \frac{1}{2^k} \sum_{S\subseteq [k]} \sqrt{\Tr(\rho_S^2)} \sqrt{\Tr(\sigma_S^2)}\\
        &\leq \frac{1}{2^k}\sum_{S \subseteq [k]} \frac{\Tr(\rho_S^2) + \Tr(\sigma_S^2)}{2} = \frac{\PT_k(\rho) + \PT_k(\sigma)}{2}.
    \end{align*}
\end{proof}

\newpage
We also provide the following lemma.
\begin{lemma}\label{lemma:helper-lemma-amgm}
    Let $x,y,z \geq 0$ with $x+y=z$ and $0\leq z\leq \frac{1}{2}$. Then, 
    \begin{equation*}
        4(x^2 + y^2) - 8(x^4+y^4) \geq 2z^2 - z^4.
    \end{equation*}
\end{lemma}
\begin{proof}
    Let's set  $y = z - x$ and minimize the LHS of the expression above with respect to $x$. By computing the critical points of the function, the quantity is minimized, when $x = y = z/2$, and, by substituting that on the expression above, we get the lower bound of the lemma.
\end{proof}

We are now ready to prove \cref{thm:gstoqma(k)=gstoqma(2)}.
\begin{proof}[Proof]
For completeness, it is easy to see that the verification protocol described in \cref{proc:stoqma(k)-to-stoqma(2)} achieves completeness $c' = \frac{1+c}{2}$. Indeed, the two Merlins have sent the same $k$-partite state, so the product test passes with certainty. Consequently, Arthur either accepts with probability 1, or accepts with the same probability as algorithm $\cA$. which is at least $c$. 

For soundness, we are going to follow the proof of Lemma 5 of \cite{HM10}.  Let us assume that Arthur receives states $\ket{\psi_1}$ and $\ket{\psi_2}$, whose maximal overlap with a product state is $1-\varepsilon_1$ and $1-\varepsilon_2$ respectively. 

Let us assume, also, that $1-\delta$ is the probability that product test accepts on $\ket{\psi_1} \otimes \ket{\psi_2}$. Then, in Lemma 5 of \cite{HM10} it is shown that the product test on $\ket{\psi_1}\otimes\ket{\psi_2}$ can be upper bounded in terms of $\PT_k(\ket{\psi_1})$ and $\PT_k(\ket{\psi_2})$ as follows.
\begin{align*}
    \PT_k(\ket{\psi_1}, \ket{\psi_2})  \leq \frac{\PT_k(\ket{\psi_1}) + \PT_k(\ket{\psi_2})}{2} \leq 1 - \frac{\varepsilon_1 + \varepsilon_2}{3} + \frac{\varepsilon_1^2+\varepsilon_2^2}{6},
\end{align*}
where at the last inequality we applied the general upper bound of \cref{thm:prod-test-sw-upper-bound}. So, we found an upper bound for the acceptance probability $1-\delta$ of item $(a)$ of \cref{proc:stoqma(k)-to-stoqma(2)}. For item $(b)$ of \cref{proc:stoqma(k)-to-stoqma(2)}, using \cref{lemma:close-overlap}, the probability of acceptance can be upper bounded by 
\begin{equation*}
    \min \bigg\{ s + \frac{\sqrt{\varepsilon_1} + \sqrt{\varepsilon_2}}{2}, 1\bigg\}.
\end{equation*}
 So, the total probability of acceptance is 
    \begin{equation*}
        s' \leq \max \bigg\{ \frac{1}{2}(1-\delta) + \frac{1}{2} \min\Big\{s+\frac{\sqrt{\varepsilon_1}+\sqrt{\varepsilon_2}}{2}, 1 \Big\}  \bigg\} \text{ s.t. } \delta \geq \frac{\varepsilon_1+\varepsilon_2}{3} - \frac{\varepsilon_1^2+\varepsilon_2^2}{6}.
    \end{equation*}
    The maximum is achieved, when 
    \begin{equation*}
        \frac{\sqrt{\varepsilon_1}+\sqrt{\varepsilon_2}}{2} = 1-s \text{ and }
        \delta =\frac{\varepsilon_1+\varepsilon_2}{3} - \frac{\varepsilon_1^2+\varepsilon_2^2}{6},
    \end{equation*}
    and then $s' \leq 1-\delta/2$. So, in order to find a lower bound for $\delta$, we apply \cref{lemma:helper-lemma-amgm} with $x=\sqrt{\frac{\varepsilon_1}{4}}, y = \sqrt{\frac{\varepsilon_2}{4}}, z=1-s$, and we get  
    \begin{equation*}
        \delta = \frac{4}{3}(x^2+y^2) - \frac{8}{3} (x^4 + y^4) \geq  \frac{2}{3} (1-s)^2 - \frac{1}{3} (1-s)^4.
    \end{equation*}
    Thus, the acceptance probability is
    \begin{equation*}
        s' \leq  1 - \frac{\delta}{2} \leq 1 - \frac{(1-s)^2}{3} + \frac{(1-s)^4}{6}  = \frac{5}{6} + \frac{1}{6}(s^2-2s)^2,
    \end{equation*}
    which concludes the proof of \cref{thm:gstoqma(k)=gstoqma(2)}.

\end{proof}

\end{document}